\documentclass[]{article}

\usepackage{amsmath,amssymb,amsfonts,amsthm,bbm}
\usepackage[ruled,vlined,longend]{algorithm2e}
\usepackage{enumitem}
\usepackage[margin=1in]{geometry}
\usepackage{natbib}
\usepackage{algorithmic}
\usepackage{graphicx}
\usepackage{textcomp}
\usepackage{xcolor}
\usepackage{subcaption}
\usepackage{comment}
\usepackage[hidelinks]{hyperref}

\newtheorem{theorem}{Theorem}
\newtheorem{lemma}{Lemma}
\newtheorem{assump}{Assumption}
\newtheorem{definition}{Definition}
\newcommand{\ones}{{{\mathbbm{1}}}}
\newcommand{\ind}[1]{\ones_{\{#1\}}}

\newcommand{\bE}{\mathbb{E}}

\newcommand{\cS}{\mathcal{S}}

\newcommand{\cD}{\mathcal{D}}

\newcommand{\GQ}{\mathrm{GQ}}
\newcommand{\LR}{\mathrm{LR}}
\newcommand{\off}{\mathrm{off}}
\newcommand{\h}[1]{\widehat{#1}}
\newcommand{\til}[1]{\widetilde{#1}}
\newcommand{\RNum}[1]{\uppercase\expandafter{\romannumeral #1\relax}}

\newcommand{\vertiii}[1]{{\left\vert\kern-0.25ex\left\vert\kern-0.25ex\left\vert #1 
		\right\vert\kern-0.25ex\right\vert\kern-0.25ex\right\vert}}
\DeclareMathOperator*{\argmax}{arg\,max}
\DeclareMathOperator*{\argmin}{arg\,min}
\usepackage{graphicx} 
\usepackage{moreverb}

\newcommand\BibTeX{{\rmfamily B\kern-.05em \textsc{i\kern-.025em b}\kern-.08em
T\kern-.1667em\lower.7ex\hbox{E}\kern-.125emX}}

\begin{document}

    \title{Nonparanormal Graph Quilting with Applications to Calcium Imaging}
    \author{Andersen Chang$^{*1}$ \and Lili Zheng$^{*2}$ \and Gautam Dasarthy$^3$ \and Genevera I. Allen$^{2,4,5,6}$}

    \date{\small%
        $^1$Department of Neuroscience, Baylor College of Medicine\\%
        $^2$Department of Electrical and Computer Engineering, Rice University\\%
        $^3$School of Electrical, Computer and Energy Engineering, Arizona State University\\%
        $^4$Department of Computer Science, Rice University\\%
        $^5$Department of Statistics, Rice University\\%
        $^6$Jan and Dan Duncan Neurological Research Institute, Texas Children’s Hospital
    }

  \maketitle
  
\footnotetext[1]{Co-first authors on this manuscript.}

\begin{abstract}
Probabilistic graphical models have become an important unsupervised learning tool for detecting network structures for a variety of problems, including the estimation of functional neuronal connectivity from two-photon calcium imaging data. However, in the context of calcium imaging, technological limitations only allow for partially overlapping layers of neurons in a brain region of interest to be jointly recorded. In this case, graph estimation for the full data requires inference for edge selection when many pairs of neurons have no simultaneous observations. This leads to the Graph Quilting problem, which seeks to estimate a graph in the presence of block-missingness in the empirical covariance matrix. Solutions for the Graph Quilting problem have previously been studied for Gaussian graphical models; however, neural activity data from calcium imaging are often non-Gaussian, thereby requiring a more flexible modeling approach. Thus, in our work, we study two approaches for nonparanormal Graph Quilting based on the Gaussian copula graphical model, namely a maximum likelihood procedure and a low-rank based framework. We provide theoretical guarantees on edge recovery for the former approach under similar conditions to those previously developed for the Gaussian setting, and we investigate the empirical performance of both methods using simulations as well as real data calcium imaging data. Our approaches yield more scientifically meaningful functional connectivity estimates compared to existing Gaussian graph quilting methods for this calcium imaging data set.
\end{abstract}


\maketitle

\section{Introduction} \label{sec:intro}

Probabilistic graphical models are a popular unsupervised learning technique for inference and sparse edge selection in network estimation, and are an important tool for understanding dependency structures in high-dimensional data. Graphical modeling approaches have been employed for data analysis in a wide variety of areas, including neuroscience \citep{yatsenko2015, carrillo2021, Subramaniyan2018}, genomics \citep{allenliu2013, Hartemink2000}, and sensor networks \citep{Chen2006, Dasarathy2016}. One particular research problem where graphical models are applied is in the study of functional neuronal connectivity, defined as statistical relationships between the activities of neurons in the brain, from two-photon calcium imaging data \citep{horwitz2003, fingelkurts2005}. Functional neuronal connectivity is of particular interest in the realm of neuroscience as a mechanism for describing how neuronal circuits in the brain are organized and for finding patterns in neuronal activity that underlies how information is passed between different regions of the brain \citep{feldt2011}; it may also serve well as a proxy for deriving the structural synaptic connectivity between individual neurons in the brain, as well as provide insight into the relationship between the two \citep{deco2014}. 

Due to technological limitations, in many calcium imaging experiments, the full set of neurons in a brain region of interest are never simultaneously observed. Instead, scans of functional activity are recorded in sequential, partially overlapping layers, each containing only a subset of the population of neurons \citep{greinberger2012, berens2017, pnevmatikakis2016}. This data collection scheme leads to block-missingness in the ensuing computation of the empirical covariance matrix for the functional recording data on all neurons, as the joint activity of many pairs are never observed; we demonstrate this visually in Figures \ref{fig:gqdat} and \ref{fig:gqsig}. Therefore, in order to estimate a graphical model for functional neuronal connectivity for the full set of observed neurons, the edge structure in the missing portion must be inferred using the information from the existing contemporaneous joint observations. The estimation of a graph in the presence of block missing entries in the covariance matrix is known as the Graph Quilting problem \citep{vinci2019graph}. 

Previously, several approaches have been developed in order to address the Graph Quilting problem. \citep{vinci2019graph} originally proposed the Maximum Determinant (MAD$_{\GQ}$) algorithm, which first finds an $\ell_1$-regularized maximum likelihood estimate of the graph on the observed portion of the covariance matrix with the constraint that no edges are affiliated with unobserved elements, then applies thresholding and Schur complements on the result in order to identify graph edges and a minimal superset of edges. Later, \citep{chang2022low} introduced the Low Rank Graph Quilting (LRGQ) approach, which utilizes a two-step procedure of covariance imputation under the assumption of the existence of a low-rank representation of the underlying covariance matrix, followed by the application of the graphical Lasso. Both of the aforementioned Graph Quilting procedures have shown positive results for graph imputation and edge recovery in the presence of block-wise missingness in the entries of the full covariance matrices. However, their scope is currently restricted to the case where the full underlying data follow a multivariate Gaussian distribution. In particular, as the sample covariance is the sufficient statistic for Gaussian graphical models, these methods and their theory are both developed with a focus on accurately estimating or imputing the covariance in the Gaussian setting. The empirical studies in these papers for validating their methods also focused only on data fit to standard Gaussian graphical models. Another related work, \citep{massa2013effectiveness}, also proposes and empirically studies an approach of combining multiple graphical models of subsets of variables to a full graph, which are in similar spirit to the MAD$_{\GQ}$ algorithm in \citep{vinci2019graph}, with a focus on the Gaussian setting.
\begin{figure}[t]
\begin{center}
    \begin{subfigure}[t]{0.25\linewidth}
    \centering
        \includegraphics[width=\linewidth]{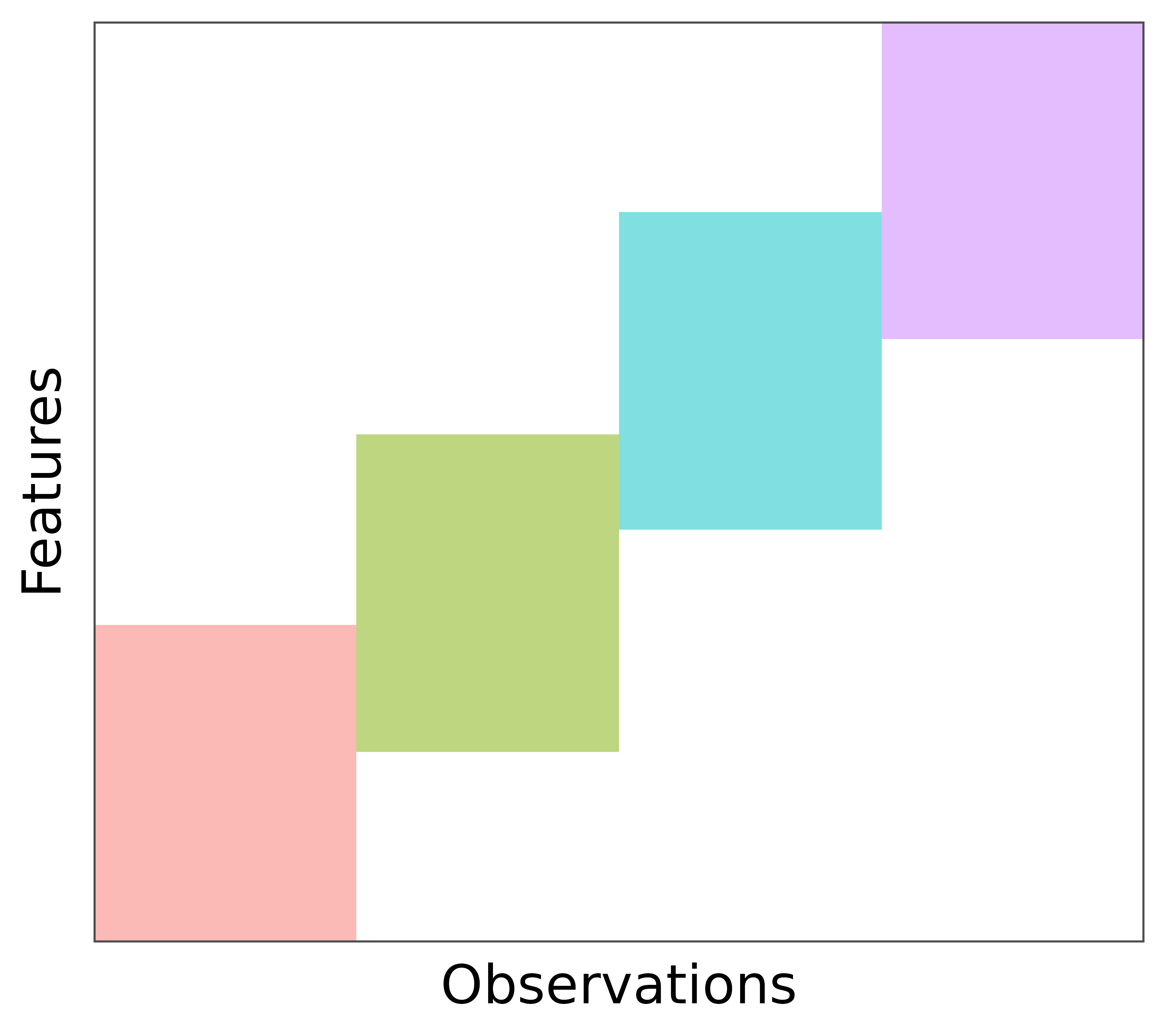}
        \caption{}
    \label{fig:gqdat}
    \end{subfigure}
    \hspace{0.5in}
    \begin{subfigure}[t]{0.25\linewidth}
    \centering
        \includegraphics[width=\linewidth]{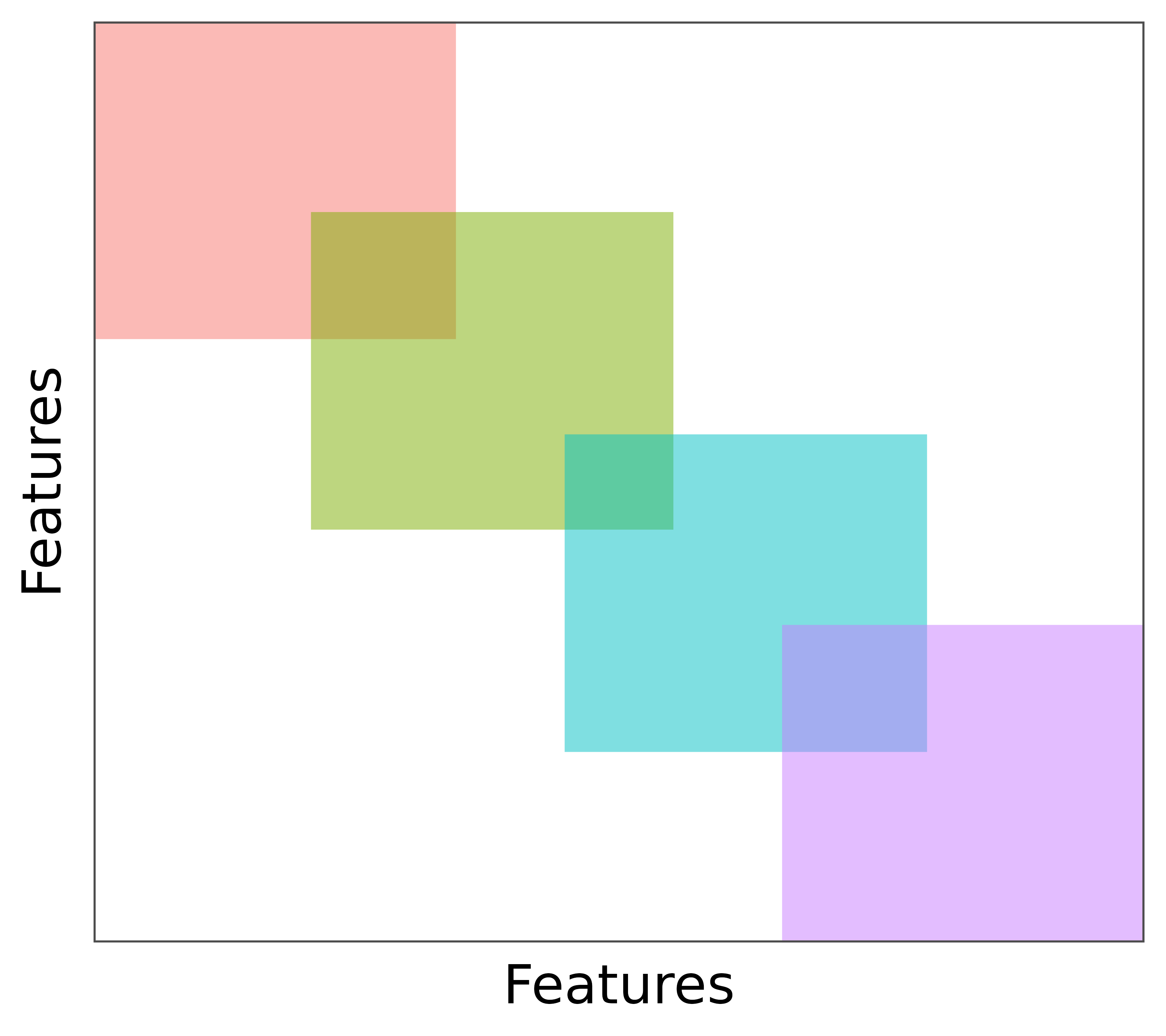}
        \caption{}
    \label{fig:gqsig}
    \end{subfigure} \\
    \begin{subfigure}[t]{0.42\linewidth}
    \centering
        \includegraphics[width=\linewidth]{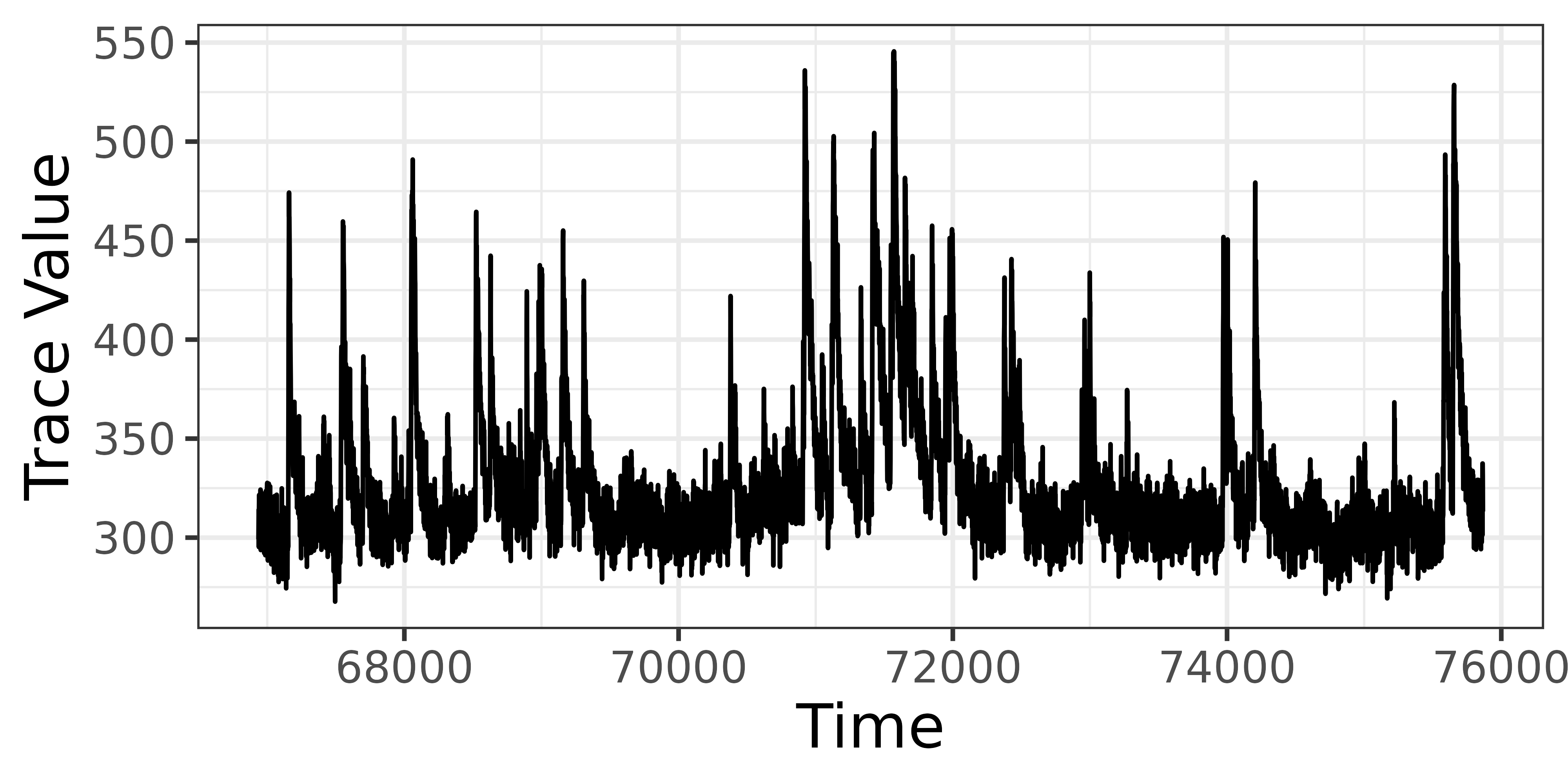}
        \caption{}
    \label{fig:flt}
    \end{subfigure}
    \begin{subfigure}[t]{0.42\linewidth}
    \centering
        \includegraphics[width=\linewidth]{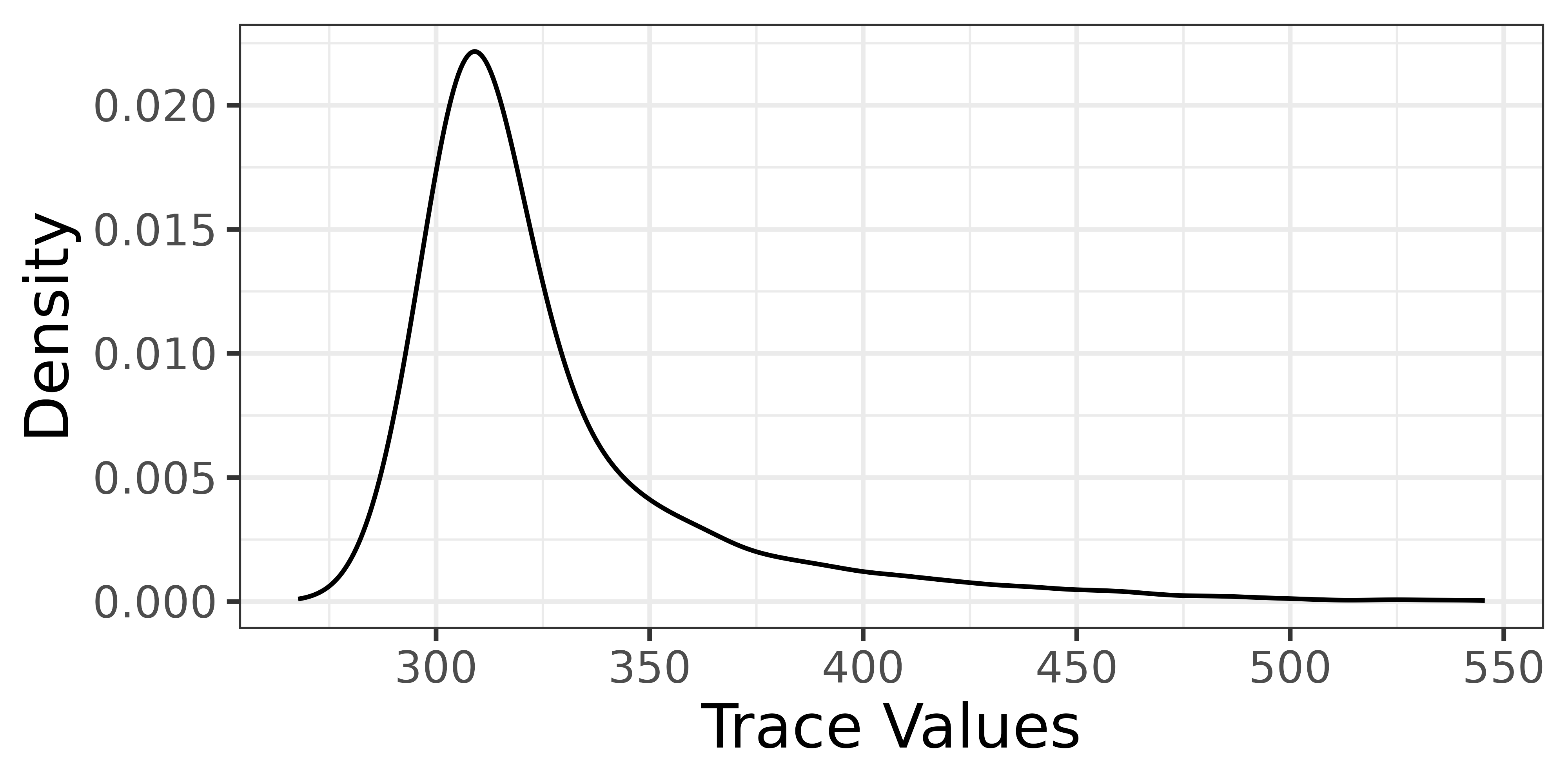}
        \caption{}
    \label{fig:fld}
    \end{subfigure}
\end{center}
\caption{\textbf{(a)}: An example of a typical schema which requires Graph Quilting. Here, we have four subsets of features in a full feature set, observed across different sessions; each rectangle represents the features and observations in a particular block. \textbf{(b)}: The corresponding incomplete empirical covariance matrix for the same four patches of nodes from (a). The parts of the covariance matrix not covered by any block are never jointly observed; graph edges in this part of the covariance must be inferred from existing entries. \textbf{(c)}: Example of raw fluorescence trace for one neuron recorded in a calcium imaging experiment, showing neuronal activity over time. \textbf{(d)}: Resulting empirical density of distribution of trace values.}
    \label{fig:fl}
\end{figure}

Functional activity data from calcium imaging, on the other hand, tends to be highly non-Gaussian. In particular, as shown in Figures \ref{fig:flt} and \ref{fig:fld}, the distribution of the fluorescence traces, which represent the firing activity of each individual neuron from calcium imaging data, is heavily right skewed with extreme positive outliers. To address this problem in the realm of graphical models, approaches which assume that the data follow a parametric distribution other than the Gaussian have been developed. For example, many have explored elliptical distributions \citep{vogel2011elliptical,finegold2011robust} and general exponential family distributions \citep{yang2015graphical,yang2018semiparametric} when the Gaussian assumption may not be appropriate. For the specific case of estimating functional neuronal connectivity, the most common alternative to a Gaussian-based method is the Poisson graphical model \citep{yang2013poisson}, which assumes that the number of spikes of each neuron across time bins follows a Poisson distribution \citep{xue2014pgm, vinci2018pgm}. The Poisson distribution is also the basis for a variety of other approaches for estimating functional connectivity outside of the graphical modeling paradigm, such as the linear-nonlinear model \citep{pillow2008, stevenson2008} and time series modeling with inter-spike intervals \citep{masud2011isi}. More recently, \citep{chang2021sgm} proposed a different class of graphical models specifically for functional neuronal connectivity based on the Subbotin distribution in order to capture conditional dependencies between extreme values in the activity traces, which directly represent neuronal activity. Nonlinear correlation methods, which employ information theory metrics such as mutual information and joint entropy, have also been used to derive functional neuronal connectivity from calcium imaging data \citep{garofolo2009mi,stetter2012}. 

While the aforementioned non-Gaussian functional neuronal connectivity models have shown encouraging results for the calcium imaging application, they may not be naturally suitable for the Graph Quilting problem. In order to leverage existing ideas \citep{vinci2019graph,chang2022low} for solving this problem, we still hope to operate on some pairwise similarity matrix like the empirical covariance calculated from the observations rather than the raw data itself. Therefore, we consider one particular alternative to the Gaussian graphical model that can be applied to the Graph Quilting problem in the non-Gaussian setting and which has been used for the analysis of functional neuronal connectivity from calcium imaging data: the nonparanormal graphical model \citep{liu2009npn,10.1214/10-AOAS397}. This model assumes a Gaussian copula for the joint distribution of all features while allowing arbitrary marginal distributions for each feature. Although some prior methods \citep{liu2009npn} estimate a univariate transform for each feature before graph estimation \citep{huge}, it has also been shown that nonparanormal graphical models can be estimated by applying the graphical Lasso \citep{originalglasso} on a transformed rank-based correlation matrix \citep{liu2012high, harris2013npn, he2017npn, xue2012}; thus, this particular procedure could be utilized in the Graph Quilting setting. 

In this paper, we study two potential Graph Quilting techniques for nonparanormal graphical models for graph estimation with non-Gaussian data under block-missingness in the pairwise observation set. Our methods utilize the MAD$_{\GQ}$ and LRGQ approaches of \citep{vinci2019graph} and \citep{chang2022low}, as applied to rank-based correlation matrices; we call the nonparanormal adaptations MAD$_{\GQ}$-NPN and LRGQ-NPN, respectively. While these previous works have studied the nonparanormal graphical model using rank-based correlation matrices, such as those mentioned above, ours is the first to consider the performance of these methods in the presence of block-missingness in the empirical covariance matrix. In particular, because the missingness pattern is highly structured rather than random and the number of observed entries is relatively sparse, the performance of nonparanormal Graph Quilting will not follow directly from prior results on nonparanormal graphical models. Therefore, we explore below the nonparanormal Graph Quilting approaches both from a theoretical and empirical perspective in order to ascertain whether they are appropriate for the Graph Quilting setting.

The rest of this paper is structured as follows. In Section 2, we describe the MAD$_{\GQ}$-NPN and LRGQ-NPN procedures for nonparanormal Graph Quilting. We then show in Section 3 the conditions under which the MAD$_{\GQ}$-NPN achieves exact graph recovery for observed node pairs and minimal superset recovery for unobserved node pairs. In Section 4, we present simulation studies which compare the performances of the MAD$_{\GQ}$-NPN procedure and one of the LRGQ-NPN algorithms on data from non-Gaussian parametric distributions. Lastly, in Section 5, we investigate the efficacy of the nonparanormal Graph Quilting methods for estimating functional neuronal connectivity networks on real-world calcium imaging data sets in comparison to each other as well as to the Graph Quilting methods with Gaussian assumptions.

\section{Nonparanormal Graph Quilting} \label{sec:meth}

\subsection{Problem Set-up} 
We consider the nonparanormal graphical model \citep{liu2009npn}, where each sample vector $X_i\in \mathbb{R}^p$ follows a nonparanormal distribution $\mathrm{NPN}_p(f,\Sigma)$, formally defined in Definition \ref{def:npn}.
\begin{definition}\label{def:npn}
Let $f=(f_1,\dots,f_p)$ be an ordered list of monotone univariate functions and let $\Sigma\in \mathbb{R}^{p\times p}$ be a positive-definite correlation matrix. Then a random vector $X=(X_1,\dots,X_p)^\top$ follows the nonparanormal distribution $\mathrm{NPN}_p(f,\Sigma)$ if $(f_1(X_1),\dots,f_p(X_p))^\top\sim\mathcal{N}(0,\Sigma)$.
\end{definition}
Define the precision matrix of the latent Gaussian vector $(f_1(X_1),\dots,f_p(X_p))$ as $\Theta = \Sigma^{-1}$; from this, we consider the following the graph structure that encodes the conditional dependence relationship among $X$: $\mathcal{G} = (V,E)$, $V=[p]$,$E=\{(j,k): j,k\in [p], \Theta_{j,k}\neq 0\}$. The primary interest is to recover the unknown graph structure, or equivalently, the non-zero patterns in the off-diagonal elements of $\Theta$. When one observes i.i.d. samples $X_1,\dots,X_n$ from this nonparanormal model, rank-based methods \citep{liu2012high,harris2013npn} have been proposed to learn the graph structure with selection consistency guarantees.

In the Graph Quilting setting, however, we do not have access to all features for each sample. Instead, we observe $K$ partially overlapping blocks $V_1,\dots,V_K$, each of size $p_k = |V_k|<p$. The corresponding observed data matrix is denoted by $X^{(k)}\in \mathbb{R}^{n_k\times p_k}$ where $n_k$ is the sample size for block $k$, and our goal is to learn the graph structure of all $p$ features from $\{X^{(k)}\}_{k=1}^K$. We define the jointly observed feature pairs as $O = \{(i,j): \exists 1\leq k\leq K, i,j\in V_k\}\subset[p]\times[p]$, and let $O^c=[p]\times [p]\backslash O$ denote the feature pairs that has no joint measurement. Thus, we need to infer the graph structure of the missing portion $O^c$ of the corresponding covariance matrix from the known pairwise observations in $O$. 

\subsection{Nonparanormal Graph Quilting Methods}\label{sec:method}
In this section, we extend two prior graph quilting approaches, the MAD$_{\GQ}$ (MAximum Determinant graph-quilting) approach \citep{vinci2019graph} and the LRGQ (Low-rank Graph Quilting) approach \citep{chang2022low}, from the Gaussian graphical model setting to the nonparanormal setting. We first note that both approaches take an estimate of the covariance/correlation $\Sigma_{i,j}$ of the Gaussian variable pairs $(i,j)$ in $O\subset [p]\times [p]$ as the input. In the nonparanormal setting, however, we need the covariance of the latent Gaussian variables $f_1(X_1),\dots,f_p(X_p)$, which can be estimated using rank-based correlations. For our work, we consider two rank-based correlations: Spearman's rho and Kendall's tau. For each block $1\leq k\leq K$, let $r^{(k)}_{i,j}$ be the rank of $X^{(k)}_{i,j}$ among $X^{(k)}_{1,j},\dots, X^{(k)}_{n_k,j}$. Also, let $\bar{r}^{(k)}_j = \frac{1}{n_k}\sum_{i=1}^{n_k} r^{(k)}_{i,j}=\frac{n_k+1}{2}$. From these, we compute Spearman's rho $\widehat{\rho}^{(k)}\in \mathbb{R}^{p_k\times p_k}$ and Kendall's tau $\widehat{\tau}^{(k)}\in \mathbb{R}^{p_k\times p_k}$ correlations as follows:
\begin{equation}
\begin{split}
    \widehat{\rho}^{(k)}_{j,l} &= \frac{\sum_{i=1}^{n_k}(r^{(k)}_{i,j} - \bar{r}^{(k)}_{j})(r^{(k)}_{i,l} - \bar{r}^{(k)}_{l})}{\sqrt{\sum_{i=1}^{n_k}(r^{(k)}_{i,j} - \bar{r}^{(k)}_{j})^2\sum_{i=1}^{n_k}(r^{(k)}_{i,l} - \bar{r}^{(k)}_{l})^2}},\\
    \widehat{\tau}^{(k)}_{j,l} &= \frac{2}{n_k(n_k-1)}\sum_{1\leq i<i'\leq n_k}\mathrm{sign}((r^{(k)}_{i,j} - r^{(k)}_{i',j})(r^{(k)}_{i,l} - r^{(k)}_{i',l})).
\end{split}
\end{equation}
To obtain a $p\times p$ correlation matrix, we combine the $K$ rank correlations $\widehat{\rho}^{(k)}$, $\widehat{\tau}^{(k)}$ together. Specifically, for any index $j\in [p]$, if $j\in V_k$, let $j_k$ be its corresponding index in $V_k$. That is, $(V_k)_{j_k} = j$. Then we formally define $\widehat{\rho}, \widehat{\tau}\in \mathbb{R}^{p\times p}$ as follow: for $(j,l)\in O,$ 
\begin{equation}\label{eq:rho_tau}
    \widehat{\rho}_{j,l} = \frac{\sum_{k=1}^K\ind{j,l\in V_k}\widehat{\rho}^{(k)}_{j_k,l_k}}{\sum_{k=1}^K\ind{j,l\in V_k}}, \quad \widehat{\tau}_{j,l} = \frac{\sum_{k=1}^K\ind{j,l\in V_k}\widehat{\tau}^{(k)}_{j_k,l_k}}{\sum_{k=1}^K\ind{j,l\in V_k}};
\end{equation}
otherwise, $\widehat{\rho}_{j,l}=\widehat{\tau}_{j,l}=0$.
As in \citep{liu2012high}, since both Spearman's rho and Kendall's tau rank correlations
are biased for the population correlation, we apply elementwise sine function transformations to $\widehat{\rho}$ and $\widehat{\tau}$: for any $(j,l)\in [p]\times [p]$,
\begin{equation}\label{eq:Sigma_rank}
    \widehat{\Sigma}^{(\rho)}_{j,l} = 2\sin\left(\frac{\pi}{6}\widehat{\rho}_{j,l}\right),\quad \widehat{\Sigma}^{(\tau)}_{j,l} = \sin\left(\frac{\pi}{2}\widehat{\tau}_{j,l}\right).
\end{equation}
Both initial observed correlation matrices take the value of zero in $O^c$, 
and $\widehat{\Sigma}^{(\rho)}_O, \widehat{\Sigma}^{(\tau)}_{O}$ serve as estimates for the population correlation $\Sigma_{O}$. 
In the following, we describe the MAD$_{\GQ}$-NPN and LRGQ-NPN approaches based on $\widehat{\Sigma}^{(\rho)}_O$ or $\widehat{\Sigma}^{(\tau)}_O$. 
\paragraph{MAD$_{\GQ}$-NPN approach} This approach is based on a partially observed likelihood, following the MAD$_{\GQ}$ method in \citep{vinci2019graph}. In a nutshell, the algorithm consists of two steps: (i) first estimating the edge set within $O$ by minimizing an $\ell_1$-regularized log-likelihood loss under the constraint of having no edges in $O^c$, followed by a thresholding step that eliminates the bias effect of this artificial constraint; (ii) then using the Schur complement to estimate a superset of the edges in $O^c$, based on potential graph distortions in $O$ caused by edges in $O^c$. The full procedure is summarized in Algorithm \ref{alg:MADgq}. Similar to the original MAD$_{\GQ}$ method, Algorithm \ref{alg:MADgq} involves three tuning parameters: the $\ell_1$ regularization parameter $\Lambda\in \mathbb{R}^{p\times p}$ in \eqref{eq:MADgqLasso}, and two thresholding parameters $\tau_1$ and $\tau_2$ for obtaining edge sets in $O$ and $O^c$. For $\Lambda$, one can either simply let $\Lambda_{j,l}=\lambda$ for all $1\leq j,\,l\leq p$, or choose $\Lambda_{j,l}=C_0\sqrt{\frac{\log p}{n_{j,l}}}$ where $n_{j,l}$ is the joint sample size of node $j,\,l$ together. The specific value of $\lambda$ or the scaling factor $C_0$ can then be chosen based on the extended Bayesian information criterion \citep{foygel2010extended,gao2012tuning}, which is commonly used in the graphical model literature and computationally efficient. For the first thresholding parameter $\tau_1$ that defines the edge set $\h{E}_O\subset O$, we can use the stability selection approach \citep{liu2010stability} by examining the stability of $\h{E}_O$ when running step 1 and 2 on randomly subsampled data. For the second thresholding parameter $\tau_2$, motivated by our theoretical results presented in Section \ref{sec:theory}, we can let $\tau_2=c\sqrt{\frac{\log p}{n_k}}$ for a small constant $c>0$. Throughout our empirical studies in this paper, we use $c=0.05$ which turns out to work reasonably well in different settings.

Under certain assumptions on the graph structure and edge strength, \citep{vinci2019graph} shows that, in the Gaussian setting, the original MAD$_{\GQ}$ approach with sample covariance as the input is guaranteed to recover the graph structure in $O$ and find a minimum edge superset in $O^c$ (see Definition \ref{def:minimal_superset}). As we will show in Section \ref{sec:theory}, such results can also be extended to our Algorithm \ref{alg:MADgq} in the nonparanormal setting.

\begin{algorithm}[t]
\noindent{\textbf{Input}}: Block indices $\{V_k\}_{k=1}^K$, Data sets $\{X^{(k)}\in \mathbb{R}^{n_k\times p_k}\}_{k=1}^K$, 
thresholding parameters $\tau_1,\,\tau_2>0$ 
\vspace{-1mm}
 \begin{enumerate}[leftmargin=*]
    \item Compute the rank-based correlation matrix $\widehat{\Sigma}_O$ as $\widehat{\Sigma}^{(\tau)}_O$ or $\widehat{\Sigma}^{(\rho)}_O$ in \eqref{eq:Sigma_rank}.
    \item Solve the MAD$_{\text{GQlasso}}$-NPN problem
    \begin{equation}\label{eq:MADgqLasso}
        \widehat{\widetilde{\Theta}}=\argmax_{\Theta\succ 0,\Theta_{O^c}=0}\log\det\Theta-\langle\h{\Sigma}_O,\Theta_O\rangle+\|\Lambda\circ \Theta\|_{1,\text{off}}.
    \end{equation}
    \item Find the edge set in $O$: $\widehat{E}_O = \{(i,j)\in O: i\neq j,|\h{\widetilde{\Theta}}_{i,j}|>\tau_1\}$
    \item For $k=1,\dots,K$, compute the Schur complement
    \begin{equation}\label{eq:MADgqSchur}
        \h{\widetilde{\Theta}}^{(k)} = \h{\widetilde{\Theta}}_{V_k,V_k} - \h{\widetilde{\Theta}}_{V_k,V_k^c}\h{\widetilde{\Theta}}_{V_k^c,V_k^c}^{-1}\h{\widetilde{\Theta}}_{V_k^c,V_k}.
    \end{equation}
    \item Find the node set
    \begin{equation}\label{eq:superset_node}
        \h{W}_{\tau_1,\tau_2}=\{i\in V: \forall k~s.t.~i\in V_k,~\exists j\neq i,~\tau_2<|\h{\widetilde{\Theta}}^{(k)}_{i,j}|<\tau_1\}
    \end{equation}
    \item Obtain a superset of edges in $O^c$: $\h{E}_{O^c} = O^c\cap(\h{W}_{\tau_1,\tau_2}\times \h{W}_{\tau_1,\tau_2})$.
\end{enumerate}
\vspace{-2mm}
 \textbf{Output}: $\h{E} = \h{E}_O \cup \h{E}_{O^c}$
\caption{MAD$_{\GQ}$-NPN Algorithm}
 \label{alg:MADgq}
\end{algorithm}

\paragraph{LRGQ-NPN approach}
Another approach we consider is extending the LRGQ method proposed by \citep{chang2022low}, which makes use of the potential low-rankness in covariance matrices. This approach is motivated by the commonly seen approximate low-rankness in many neuroscience data sets and exhibits promising performance in these applications. The LRGQ approach is a two-step procedure that first completes the covariance estimate from the partially observed sample covariance in $O$ using appropriate low-rank matrix completion methods, then uses the completed covariance matrix as input to the graphical Lasso algorithm to produce a graph estimate. For the imputation step, \citep{chang2022low} introduces three different imputation methods, including a block-wise SVD method (BSVDgq), a nuclear norm minimization method (NNMgq), and a non-convex gradient descent method for low-rank factorization (LRFgq). In the nonparanormal setting, we propose to apply the aforementioned two-step approach on a rank-based correlation matrix instead of the sample covariance; we summarize the full procedure in Algorithm \ref{alg:LRgq}.  Below, throughout our empirical investigation for the LRGQ-NPN approach, we focus on the BSVDgq approach in \citep{chang2022low} for the imputation step. More details on its implementation are included in the Appendix. Algorithm \ref{alg:LRgq} also involves two tuning parameters: the rank $r$ and the regularization parameter $\Lambda$. Similar to \citep{chang2022low} and other prior works on low-rank matrix completion, one can choose the appropriate rank $r$ using the Bayesian information criterion (BIC) \citep{burnham2004multimodel}. While for $\Lambda$, same as the MAD$_{\GQ}$ procedure, we can let each entry $\Lambda_{j,k}$ be the same or scale proportionally to $\sqrt{\frac{\log p}{n_{j,k}}}$. The specific values of $\Lambda$ can then be chosen based on the stability criterion \citep{liu2010stability}.

As has been pointed out in many prior works on graphical models \citep{ravikumar2011high,liu2012high}, the consistency of graphical Lasso hinges on sufficiently accurate covariance estimation in terms of each entry, which implies that the two-step procedure is guaranteed to give a consistent graph selection provided that the imputation step leads to small $\|\h{\Sigma}^{\LR} - \Sigma\|_{\infty}$. The BSVDgq method in particular has been shown in \citep{chang2022low} to achieve sufficient imputation accuracy that eventually leads to graph selection consistency in the Gaussian setting. However, such theoretical results are based on delicate spectral analysis of sample covariance matrices, which is extremely challenging to be extended to the non-linear rank-based correlations. We leave the theoretical investigation for the LRGQ-NPN approach as future work and instead focus on empirical validation here instead.


\begin{algorithm}[t]
\noindent{\textbf{Input}}: Block indices $\{V_k\}_{k=1}^K$, Data sets $\{X^{(k)}\in \mathbb{R}^{n_k\times p_k}\}_{k=1}^K$, 
rank of full covariance matrix $r$. 
\vspace{-1mm}
 \begin{enumerate}[leftmargin=*]
    \item Compute the rank-based correlation matrix $\widehat{\Sigma}_O$ as $\widehat{\Sigma}^{(\tau)}_O$ or $\widehat{\Sigma}^{(\rho)}_O$ in \eqref{eq:Sigma_rank}.
    \item Obtain imputed covariance matrix $\h{\Sigma}^{\LR}$ using low-rank covariance completion methods.
    \item Apply the graphical Lasso to the imputed full covariance matrix $\h{\Sigma}^{\LR}$ in order to obtain the estimated graph 
    $$\widehat{\Theta}^{\LR} = \argmin_{\Theta \succ 0}  \, \log\det\Theta-\langle\h{\Sigma}^{\LR},\Theta\rangle+\|\Lambda\circ \Theta\|_{1,\text{off}}$$
\end{enumerate}
\vspace{-2mm}
 \textbf{Output}: $\h{E} = \{(i,j):i\neq j\in [p],\h{\Theta}^{\LR}_{i,j}\neq 0\}$
\caption{LRGQ-NPN Algorithm}
 \label{alg:LRgq}
\end{algorithm}

\section{Edge Recovery of MAD\texorpdfstring{$_{\GQ}$}{Lg}-NPN} \label{sec:theory}
Although the idea of using rank-based correlations to substitute the sample covariance / Pearson correlation matrix is straightforward, it is not clear if this idea would succeed in the nonparanormal graph quilting setting. In this section, we examine the theoretical properties of the MAD$_{\GQ}$-NPN approach (Algorithm \ref{alg:MADgq}), giving an affirmative answer to this question by showing similar theoretical guarantees as those in \citep{vinci2019graph} established for Gaussian data. Specifically, we will show that under similar assumptions, Algorithm \ref{alg:MADgq} exactly recovers the edge set in $O$ and a minimal superset of edges in $O^c$. Recall the MAD$_{\text{GQlasso}}$-NPN solution and its Schur complements defined in Algorithm \ref{alg:MADgq}, here let's first define their population versions as follows:
\begin{equation}\label{eq:Theta_tilde_def}
    \widetilde{\Theta} = \argmax_{\Theta\succ 0,\Theta_{O^c}=0}\log\det \Theta - \sum_{(i,j)\in O}\Theta_{i,j}\Sigma_{i,j},
\end{equation}
$\widetilde{\Sigma} = \widetilde{\Theta}^{-1}$, and for $k=1,\dots,K$,
\begin{equation}\label{eq:MADgqSchur_population}
    \widetilde{\Theta}^{(k)} = \widetilde{\Theta}_{V_k,V_k} - \widetilde{\Theta}_{V_k,V_k^c}\widetilde{\Theta}_{V_k^c,V_k^c}^{-1}\widetilde{\Theta}_{V_k^c,V_k}.
\end{equation}
\citep{vinci2019graph} has established nice theoretical properties for graph selection if one has access to $\til{\Theta}$ and $\widetilde{\Theta}^{(k)}$. These results form the basis of our theory, and much of our analysis is devoted to showing the proximity of our finite sample solutions $\h{\til{\Theta}}$ in \eqref{eq:MADgqLasso} to its population counterpart $\widetilde{\Theta}$, and $\h{\tilde{\Theta}}^{(k)}$ in \eqref{eq:MADgqSchur} to $\til{\Theta}^{(k)}$. In the following, we will follow the terminology and assumptions developed for the population theory in \citep{vinci2019graph}.

As has been shown in \citep{vinci2019graph}, hard thresholding $\widetilde{\Theta}_O$ can lead to the graph recovery in $O$ under certain assumptions on graph signals. 
While for the edges in $O^c$, the exact edge set $E_{O^c}$ is not identifiable in the graph quilting setting. However, it is possible to recover a superset of $E_{O^c}$ based on the distortion in the Schur complements of $\til{\Theta}_{V_k,V_k}$ for each block $k$, created by the out-of-block edges.
Let the distortion created by the out-of-block edges in the block $k$ be $\delta^{(k)}_{i,j} = \Theta_{i,j} - \til{\Theta}^{(k)}_{i,j}$. Given the observed covariance $\Sigma_O$ and these distortion information, we define the following minimal superset of $E_{O^c}$ as the smallest possible set that include all possible edges which cannot be ruled out without further information. 
\begin{definition}[Minimal Superset  of $E_{O^c}$]\label{def:minimal_superset}
Let 
\begin{equation}
    \mathcal{D}_{\rm off}(\Sigma, O)=\big\{(i,j,k):\delta^{(k)}_{ij}\neq 0, i\neq j\big\}
\end{equation}
be the set of known distortions over the off-diagonal elements of the Schur complements $\tilde\Theta^{(1)},...,\tilde\Theta^{(K)}$, and let 
\begin{equation}\label{eq:admiss}
\mathcal{A}_{\rm off}(\Sigma,O):=\left\{\Sigma'\succ 0: ~\Sigma'_O=\Sigma_O,~ \mathcal{D}_{\rm off}(\Sigma',O)=\mathcal{D}_{\rm off}(\Sigma,O)\right\}
\end{equation}
be the set of all positive definite covariance matrices that agree with the observed $\Sigma_O$ and distortions $\mathcal{D}_{\rm off}(\Sigma, O)$. A set $\mathcal{S}_{\rm off}\subseteq O^c$ is the minimal superset of $E_{O^c}$ with respect to $\Sigma_O$ and $\mathcal{D}_{\rm off}(\Sigma,O)$ if it satisfies the following properties:
\begin{enumerate}
\item $\forall\Sigma'\in\mathcal{A}(\Sigma,O,Q)$ we have $E_{O^c}'\subseteq \mathcal{S}_{\rm off}$;
\item $\forall\mathcal{S}' \subsetneq \mathcal{S}_{\rm off}$, $\exists\Sigma'\in\mathcal{A}_{\rm off}(\Sigma,O)$ such that $E'_{O^c} \cap (\mathcal{S}_{\rm off}\setminus \mathcal{S}') \neq \emptyset$.   
\end{enumerate}
\end{definition}

We also define the following quantities that will be useful in our theory.
Let $\nu := \min_{(i,j)\in E_O}|\Theta_{i,j}|$ be the minimum signal in $E_O$, $\delta := \max_{(i,j)\in O,i\neq j}|\Theta_{i,j}-\til{\Theta}_{i,j}|$ be the maximum distortion induced by constraining $\til{\Theta}_{O^c}=0$. Also let $\psi:=\min_{(i,j,k):0<|\tilde\Theta_{ij}^{(k)}|<\delta} \min\{|\tilde\Theta_{i,j}^{(k)}|,\delta
    -|\tilde\Theta_{i,j}^{(k)}|\}$, $d = \max_i\|\Theta_{i,:}\|_0$, $\til{d} = \max_i\|\til{\Theta}_{i,:}\|_0$, $\til{s} = \|\til{\Theta}\|_{0,\off}$ be the number of non-zero off-diagonal elements of $\til{\Theta}$. Some other technical quantities include $\til{\kappa}=\frac{\lambda_{\max}(\til{\Theta})}{\lambda_{\min}(\til{\Theta})}$, $\kappa_{\til{\Sigma}} = \vertiii{\til{\Sigma}}_{\infty}=\max_j\sum_{k=1}^p|\til{\Sigma}_{j,k}|$, $\kappa_{\til{\Gamma}} = \vertiii{(\til{\Gamma}_{S,S})^{-1}}_{\infty}$, where $S$ is the support set of $\til{\Theta}$. Let $H_i = \{j: (i,j)\in O^c\}$, $N_{H_i}(i) = N(i)\cap H_i$ be the neighborhood of $i$ in $H_i$. We say that two nodes $i$ and $j$ are $V$-connected if there is a path in $V$ connecting $i$ and $j$. 
    We also require the following assumptions:
\begin{assump}[Weak distortion compared to signal]\label{assump:distortion_signal}
    We assume that the maximum off-diagonal distortion of the MAD$_{\GQ}$ solution is smaller than half the signal strength in the original precision matrix: $\delta<\frac{\nu}{2}$.
\end{assump}
\noindent As has been proven in Theorem 3.1 in \citep{vinci2019graph}, Assumption \ref{assump:distortion_signal} can be satisfied as long as $\|\Theta_{O^c}\|_{\infty}$ is bounded above by a function of the edge weights within $\Theta_{O}$; this assumption is likely to hold as long as the pairwise observation set includes all edges with strong signal strength. 
\begin{assump}\label{assump:Schur_distortion1}
     For every node $i\in V$ with $N_{H_i}(i)\neq\emptyset$, we have that for every $k$ such that $i\in V_k$, there exists at least one node $j\in V_k\setminus \{i\}$ that is $(H_i\cup\{j\})$-connected to some node in $N_{H_i}(i)$.
\end{assump}
\noindent Assumption \ref{assump:Schur_distortion1} requires that if any node $i$ has an edge in $O^c$, then for any block it belongs to, there exists a path that starts from $i$ and an edge in $O^c$, and eventually back to this block. This assumption ensures that any edge $(i,j)$ in $O^c$ would cause certain distortions in the blocks $i$ and $j$ belong to. Therefore, given the distortion set $\cD_{\off}(\Sigma,O)$, it is possible to identify the node set with edges in $O^c$. However, we can't directly compute the distortion $\delta_{i,j}^{(k)}$, a discrepancy between the unknown $\Theta_{i,j}$ and $\til{\Theta}_{i,j}$. Instead, we can only assume that such discrepancy leads to a small non-zero $\til{\Theta}_{i,j}$. This motivates our following assumption: 
\begin{assump}\label{assump:Schur_distortion2}
    If $\delta_{i,\backslash i}^{(k)}\neq 0$, then there exists $j\neq i$ such that $0<|\tilde\Theta_{ij}^{(k)}|<\delta$.
\end{assump}

\begin{assump}[Incoherence condition]\label{assump:incoh}
    Let $\Gamma = \widetilde{\Sigma}\otimes \widetilde{\Sigma}$, $S = \{(j,l): \widetilde{\Theta}_{j,l}\neq 0\}$. We assume $\max_{e\in O\cap S^c}\|\Gamma_{e,S}\Gamma_{S,S}^{-1}\|_1\leq 1-\alpha$ for some $0<\alpha\leq 1$.
\end{assump}
\noindent Similar incoherence conditions are commonly assumed in the literature of graphical models \citep{ravikumar2011high}.
\begin{assump}[Sufficient block measurements]\label{assump:V}
    The $K$ blocks cover all nodes: $\cup_{k=1}^KV_k = [p]$, and at least one off-diagonal element: $|O|>p$. 
\end{assump}
\noindent This is a mild assumption on the block measurements, which is typically satisfied by our motivating neuroscience applications.
\begin{assump}[Regularization parameter]\label{assump:lambda}
    $\Lambda_{j,l} = \frac{C_0}{\alpha}\sqrt{\frac{\log p}{\min_k n_k}}$ for all $(j,l)\in O$ and some universal constant $C_0>0$.
\end{assump}
\noindent Assumptions \ref{assump:distortion_signal}-\ref{assump:lambda} also appear in \citep{vinci2019graph} in the Gaussian graphical model setting. Note that we do not require any additional assumptions when extending the theoretical guarantees to the nonparanormal setting, while accommodating for strictly weaker assumptions on the joint distribution. 
\begin{theorem}\label{thm:main}
Suppose that Assumptions \ref{assump:distortion_signal}-\ref{assump:lambda} hold, and there exist at least one edge in the graph encoded by $\Theta$ and $\til{\Theta}$: $d,\til{d}>2$. 
Then we have the following guarantees for Algorithm \ref{alg:MADgq} with probability at least $1-\sum_kp_k^{-10}$:
\begin{itemize}
    \item \emph{\textbf{Exact recovery in $O$}}. If
\begin{equation}\label{eq:n_cond1}
    n_k> \left[\frac{C_0}{4}\kappa_{\til{\Gamma}}\left(1+\frac{8}{\alpha}\right)\left(\left(\frac{\nu}{2}-\delta\right)^{-1}+3\left(1+\frac{8}{\alpha}\right)(\kappa_{\til{\Sigma}}+\kappa^3_{\til{\Sigma}}\kappa_{\til{\Gamma}})\til{d}\right)\right]^2\log p_k,
\end{equation}
 $\delta+\varepsilon_1\leq \tau_1<\nu-\delta-\varepsilon_1$, where $\varepsilon_1 = \frac{C_0}{4}\kappa_{\til{\Gamma}}(1+\frac{8}{\alpha})\max_k\sqrt{\frac{\log p_k}{n_k}}$, then $\h{E}_O = E_O$.
    \item \emph{\textbf{Minimal superset recovery in $O^c$}}. If
\begin{equation}\label{eq:n_cond2}
n_k> C_0\kappa_{\til{\Gamma}}^2\left(1+\frac{8}{\alpha}\right)^2\left[\frac{9\til{\kappa}^4}{4\psi^2}+\frac{1}{4\lambda_{\min}(\til{\Theta})^2}\right]\min\{p+\til{s},\til{d}^2\}\log p_k,
\end{equation}
$\varepsilon_2\leq \tau_2< \psi - \varepsilon_2$, $\delta-\varepsilon_2< \tau_1\leq \nu-\varepsilon_2$, where $\varepsilon_2= \frac{3C_0}{4}\kappa_{\til{\Gamma}}(1+\frac{8}{\alpha})\til{\kappa}^2\min\{\sqrt{p+\til{s}},\til{d}\}\max_k\sqrt{\frac{\log p_k}{n_k}}$, then $\h{E}_{O^c} = \cS_{\off}$.
\end{itemize}
\end{theorem}
\noindent The proof of Theorem \ref{thm:main} can be found in the Appendix. Theorem \ref{thm:main} suggests that under the same population assumptions as the Gaussian graph quilting setting in \citep{vinci2019graph}, as long as the sample size for each block is sufficiently large: $n_k = \Omega(\til{d}^2\log p_k)$, and the thresholding parameters are appropriately chosen: $\tau_1\in [\delta+C\max_k\sqrt{\frac{\log p_k}{n_k}}, \nu - \delta-C\max_k\sqrt{\frac{\log p_k}{n_k}})$, $\tau_2\asymp \til{d}\max_k\sqrt{\frac{\log p_k}{n_k}}$, we can achieve exact recovery for edges in $O$ and construct a minimal superset for edges in $O^c$. This result is comparable to the main theory in \citep{vinci2019graph} for Gaussian graphical model, although we are considering a strictly broader distribution family. To prove Theorem \ref{thm:main}, we first extend the existing error bounds for rank-based correlation matrices \citep{liu2012high} in the full data setting to our modified correlation estimates defined in \ref{sec:method}, which is computed from $K$ semi-overlapping blocks of measurements. We then utilize this error bound to show the proximity of $\h{\til{\Theta}}$ and its Schur complements $\h{\til{\Theta}}^{(k)}$ to their population counterparts, which eventually leads to Theorem \ref{thm:main} when combined with the population theory developed in \citep{vinci2019graph}.

\section{Simulation Studies} \label{sec:sim}

\begin{figure}[t]
\begin{center}
    \begin{subfigure}[t]{0.4\linewidth}
    \centering
        \includegraphics[width=\linewidth]{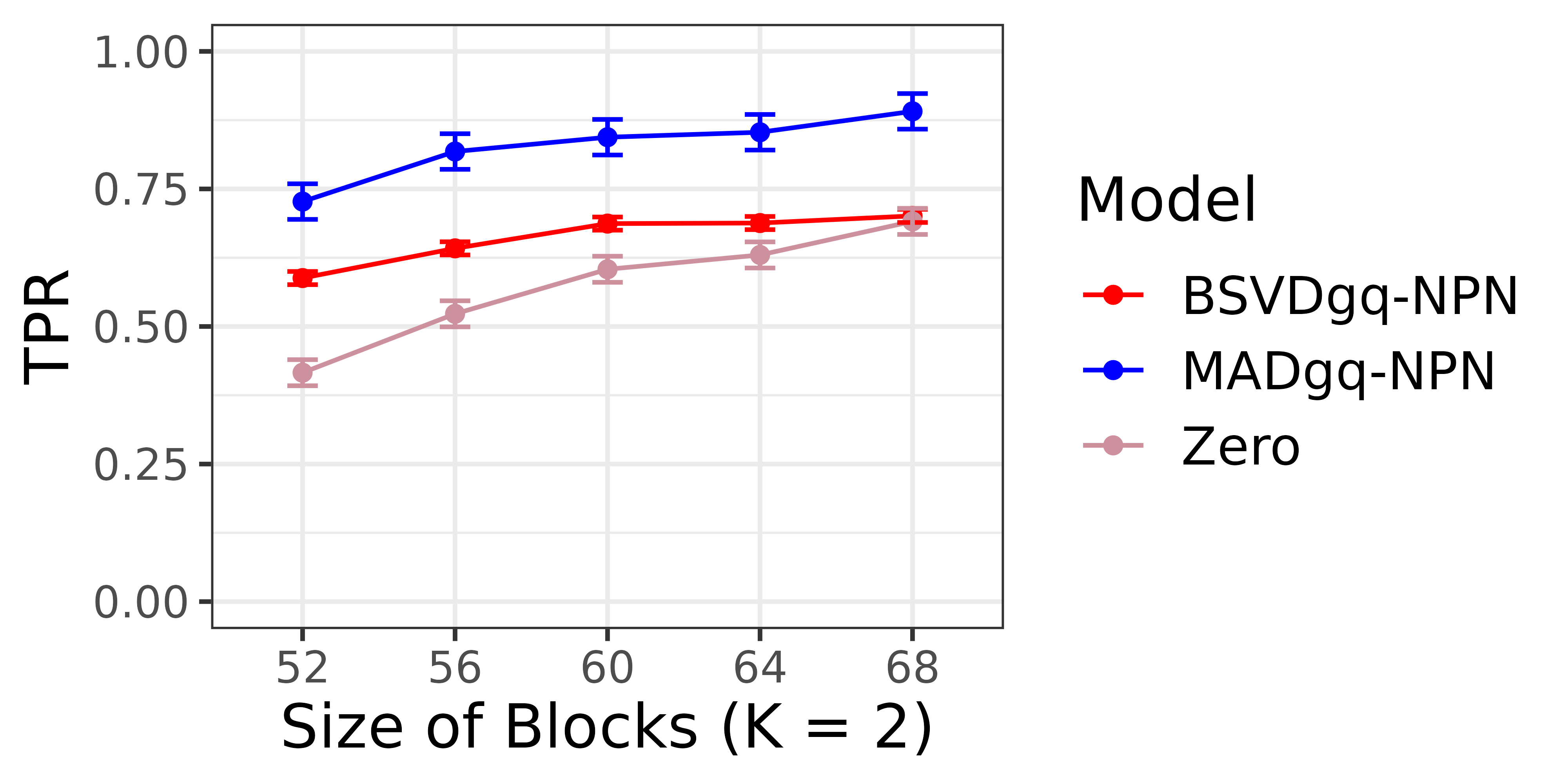}
        \caption{}
    \label{fig:gamma_o_tpr}
    \end{subfigure}%
    \begin{subfigure}[t]{0.4\linewidth}
    \centering
        \includegraphics[width=\linewidth]{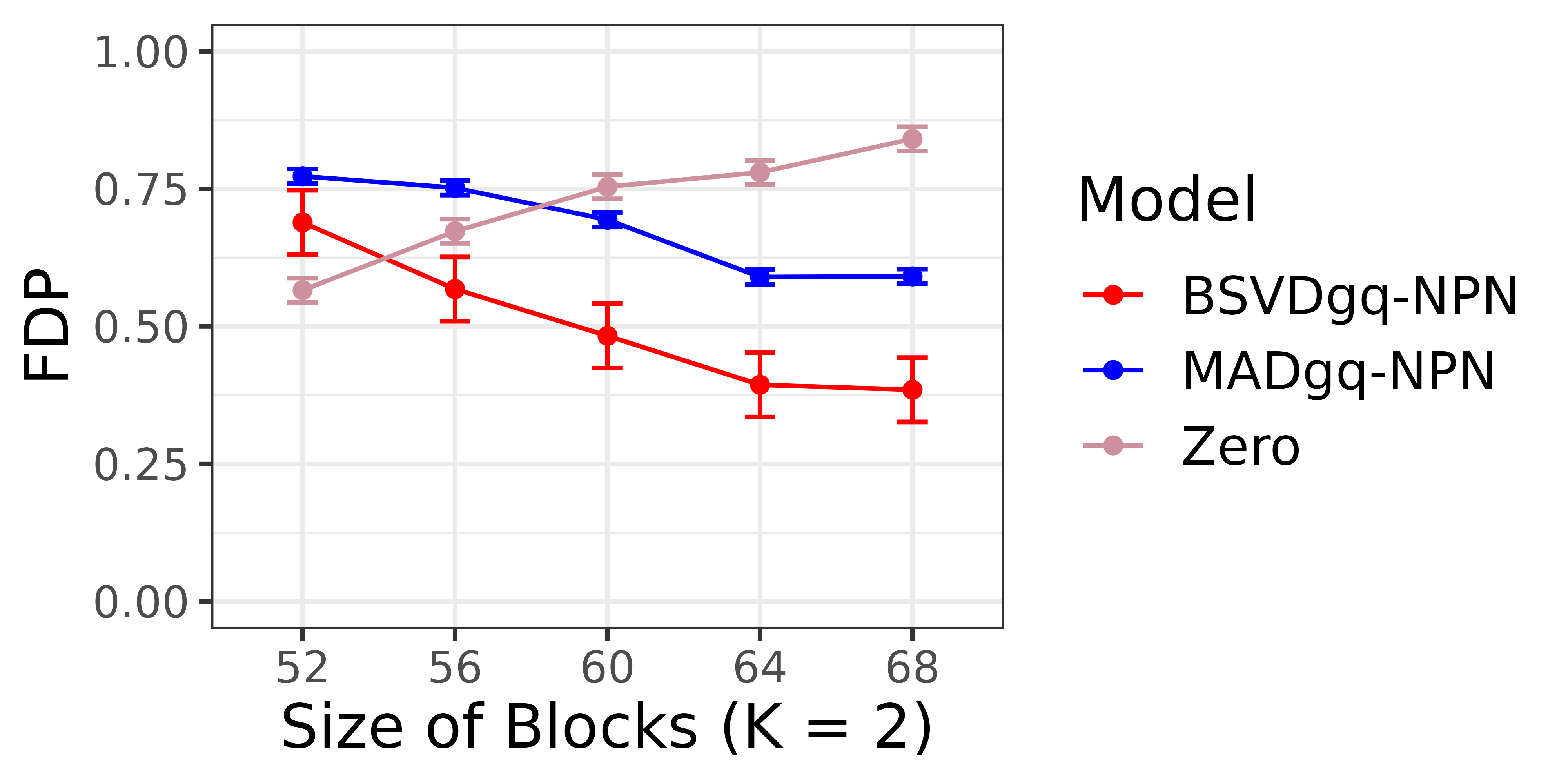}
        \caption{}
    \label{fig:gamma_o_fdp}
    \end{subfigure}
    \begin{subfigure}[t]{0.4\linewidth}
    \centering
        \includegraphics[width=\linewidth]{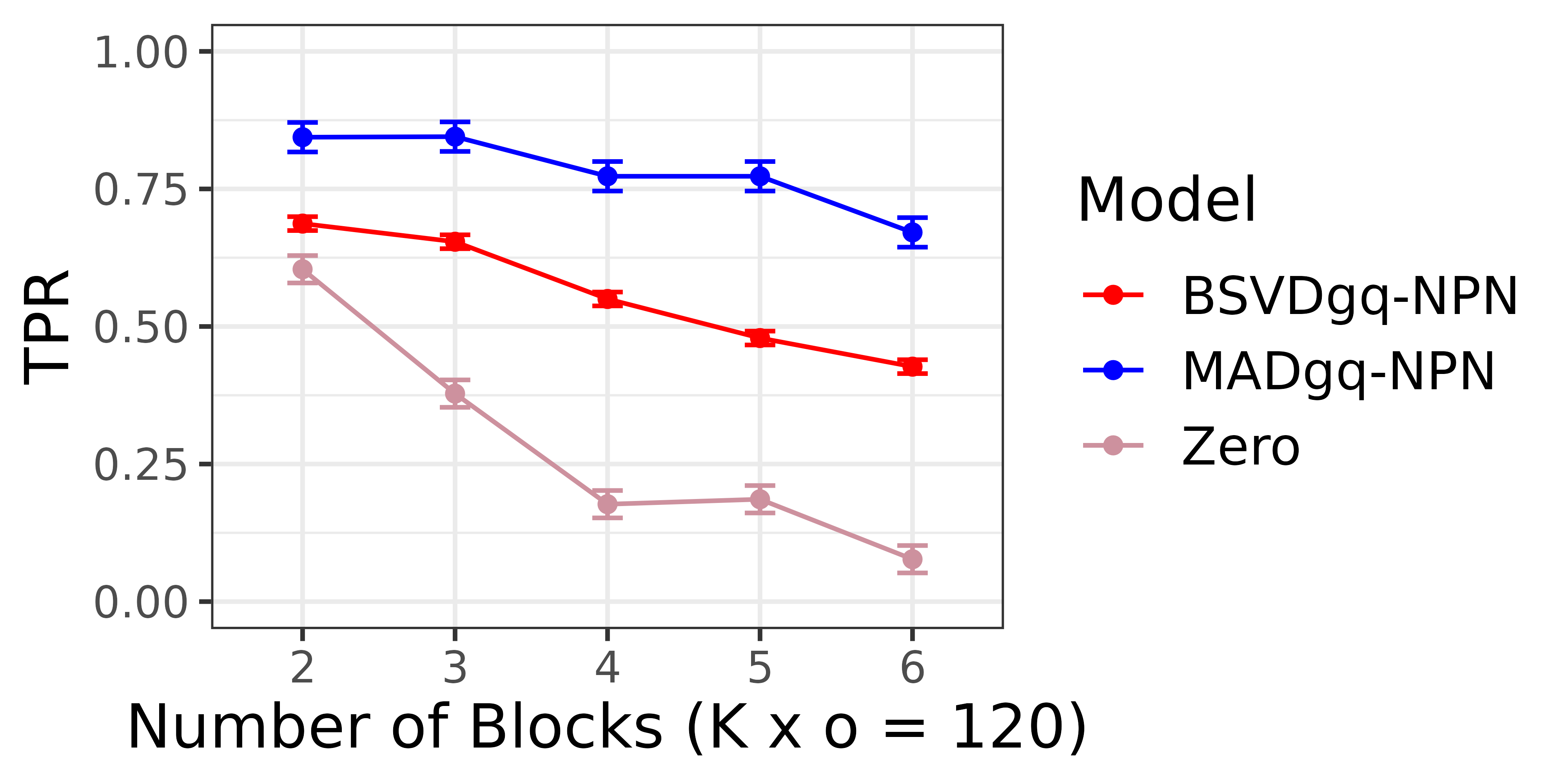}
        \caption{}
    \label{fig:gamma_k_tpr}
    \end{subfigure}
    \begin{subfigure}[t]{0.4\linewidth}
    \centering
        \includegraphics[width=\linewidth]{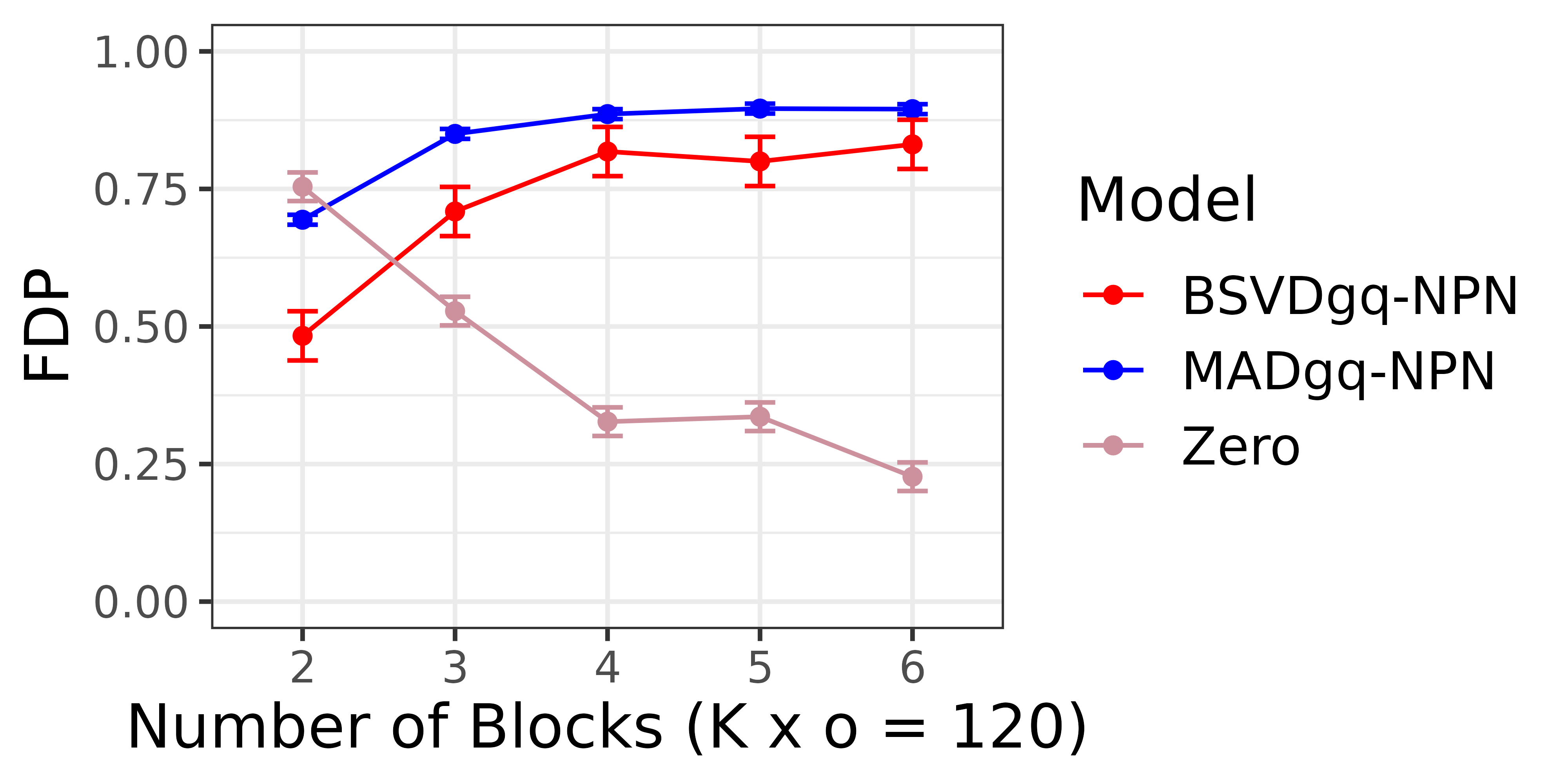}
        \caption{}
    \label{fig:gamma_k_fdp}
    \end{subfigure}%
\end{center}
    \caption{Performance of MAD$_{\GQ}$-NPN, BSVDgq-NPN, and zero imputation on simulated data from a Gamma distribution. \textbf{(a)} TPR, changing block size. \textbf{(b)} FDP, changing block size. \textbf{(c)} TPR, number of blocks. \textbf{(d)} FDP, changing number of blocks.}
    \label{fig:gamma}
\end{figure}

We now investigate the empirical performance of the nonparanormal Graph Quilting procedures outlined in Section \ref{sec:meth} on simulated data. For each of the simulation trials, we create synthetic block observation patterns by randomly ordering the features and assigning them to $K$ partially overlapping blocks of size $o$. Data is generated for each block from multivariate Gaussian data from a sparse inverse covariance matrix with a small-world structure and then perform copula transform on each column to a non-Gaussian distribution. Rank-based correlations are calculated entry-wise for each pair of features in the observation set, with correlations averaged for pairs observed in multiple blocks. Our goal is to recover the nonzero entries in the population sparse inverse covariance matrix of the unobserved Gaussian variable. Below, we compare the MAD$_{\GQ}$-NPN approach and one of the LRGQ-NPN procedures, namely the block SVD (BSVDgq-NPN) approach, along with a basic zero imputation procedure in which the missing entries of the correlation matrix are imputed as 0 before applying the graphical Lasso algorithm. The methods are evaluated using the the true positive rate (TPR) and false discovery proportion (FDP) of their respective resulting graph estimates as compared to the true underlying graph; in particular, we run 50 replications on each set of simulation parameters, with new random block assignment in each replication, and show the average and standard deviation of the TPR and FDP of edge selection for each method. Hyperparameter selection is performed via optimal F-1 score tuning with respect to the true graph.

We first test the nonparanormal Graph Quilting methods with Spearman correlation matrices calculated on data containing 100 features generated from a Gamma distribution with shape parameter 5 and scale parameter 1. For each block, 2000 observations are generated per feature. Figures \ref{fig:gamma_o_tpr} and \ref{fig:gamma_o_fdp} show the TPR and FDP of each method compared to the true underlying graph when the size of the block $o$ is 52, 56, 60, 64, and 68 while the number of blocks is held constant at 2, and Figures \ref{fig:gamma_k_tpr} and \ref{fig:gamma_k_fdp} show the TPR and FDP for $K = $ 2, 3, 4, 5 and 6 blocks while keeping the total number of observed node pairs across all blocks ($K \times o$) constant at 120. From these results, we see that both methods achieve a consistently high true positive rate for edge recovery, and generally outperforms the zero imputation method. Notably, even though the low-rank assumption inherent to the BSVDgq-NPN method is not met here, the method is still able to recover the true edges of the graph at a decently high rate. Comparing the two nonparanormal Graph Quilting methods, we see that the MAD$_{\GQ}$-NPN approach consistently has a higher TPR and FDP compared to BSVDgq-NPN; this result follows what we expect, as the former approach is designed to construct a superset of the true edges for the unobserved entries in the covariance matrix and thus will return an estimate with both more true positives and false positives compared to BSVDgq-NPN. Across different simulation parameters, we generally observe that TPR is higher and FDP is lower for the estimates from both methods when there are fewer blocks and when each block is larger, which we would expect to see as these conditions effectively provide more samples for estimation. Additionally, the difference in performance between the two nonparanormal Graph Quilting methods appear to be fairly consistent across the different block sizes and number of blocks.

We also evaluate the nonparanormal Graph Quilting methods on data simulated from a Cauchy distribution with mean 0 and scale parameter 3 with 2000 observations and 100 features, from which a Kendall correlation matrix is calculated and used as the input to the graphical Lasso. For this particular experiment, we consider the case where the underlying covariance matrix is approximately low-rank; this is generated via the spiked covariance model \citep{johnstone2001}, and we also enforce a small-world graph structure in its inverse. Figures \ref{fig:cauchy_o_tpr} and \ref{fig:cauchy_o_fdp} compare the MAD$_{\GQ}$-NPN and BSVDgq-NPN methods in terms of TPR and FDR for edge selection for $K = 2$ blocks and varying block size $o$ is 52, 56, 60, 64, and 68, and Figures \ref{fig:cauchy_k_tpr} and \ref{fig:cauchy_k_fdp} compare the same methods for $K = $ 2, 3, 4, 5 and 6 blocks with a constant total number of observations ($K \times o$) of 120. Both nonparanormal Graph Quilting methods perform well in terms of true positive rate for edge recovery here as well, and both considerably outperform the zero imputation approach. As opposed to the previous simulation study in which the MAD$_{\GQ}$-NPN method outperformed the BSVDgq-NPN method in terms of selecting true edges in the graph, we see in this case that the latter outperforms the former for both the TPR and FDP metrics. The relative performance of the two nonparanormal Graph Quilting methods matches what we would expect from this particular simulation setting, as the structure of the full true underlying covariance matrix more closely match the model set-up of the BSVDgq-NPN method, and also aligns with comparative results seen from \citep{chang2022low} comparing Graph Quilting methods under a Gaussian assumption. Even in this case, though, the MAD$_{\GQ}$-NPN method still recovers the true edges of the underlying graph reasonably well. These results show that the choice of which of the nonparanormal Graph Quilting approach to apply for a problem will depend whether a low-rank assumption makes sense in the particular scientific context. Additionally, as above, we observe that performance generally improves with larger blocks and with a fewer total number of blocks.

\begin{figure}[t]
\begin{center}
    \begin{subfigure}[t]{0.4\linewidth}
    \centering
        \includegraphics[width=\linewidth]{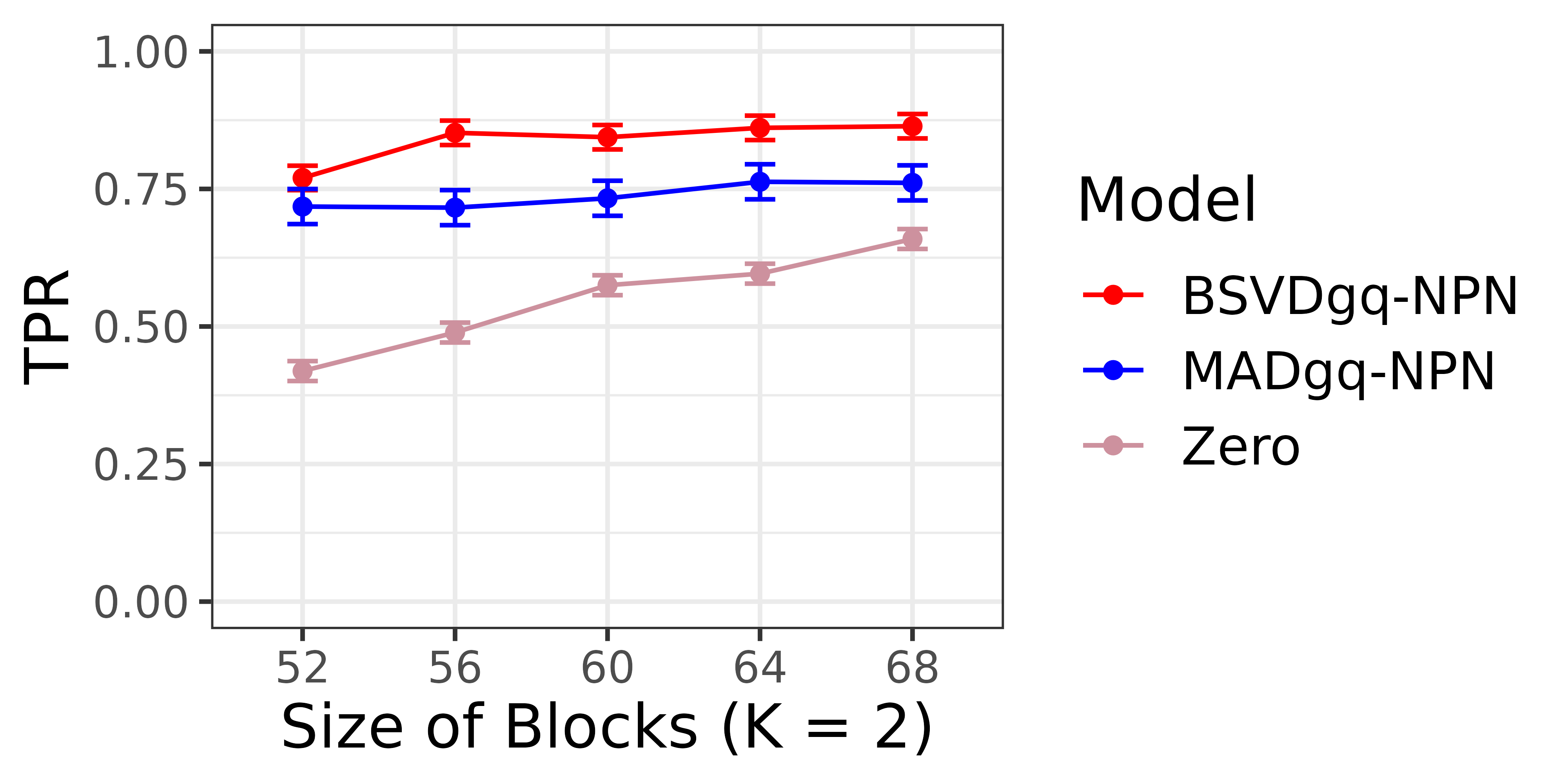}
        \caption{}
    \label{fig:cauchy_o_tpr}
    \end{subfigure}%
    \begin{subfigure}[t]{0.4\linewidth}
    \centering
        \includegraphics[width=\linewidth]{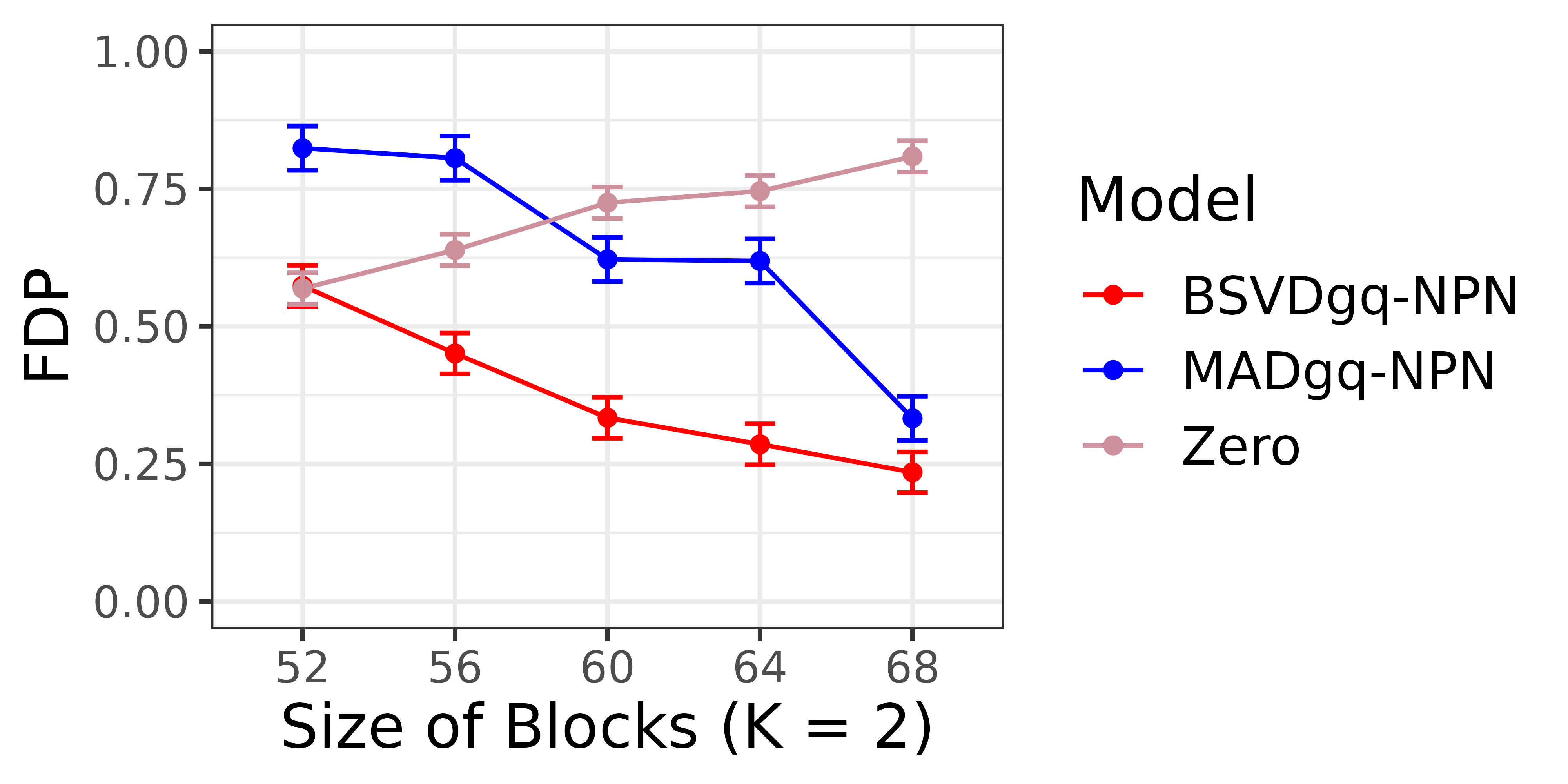}
        \caption{}
    \label{fig:cauchy_o_fdp}
    \end{subfigure}
    \begin{subfigure}[t]{0.4\linewidth}
    \centering
        \includegraphics[width=\linewidth]{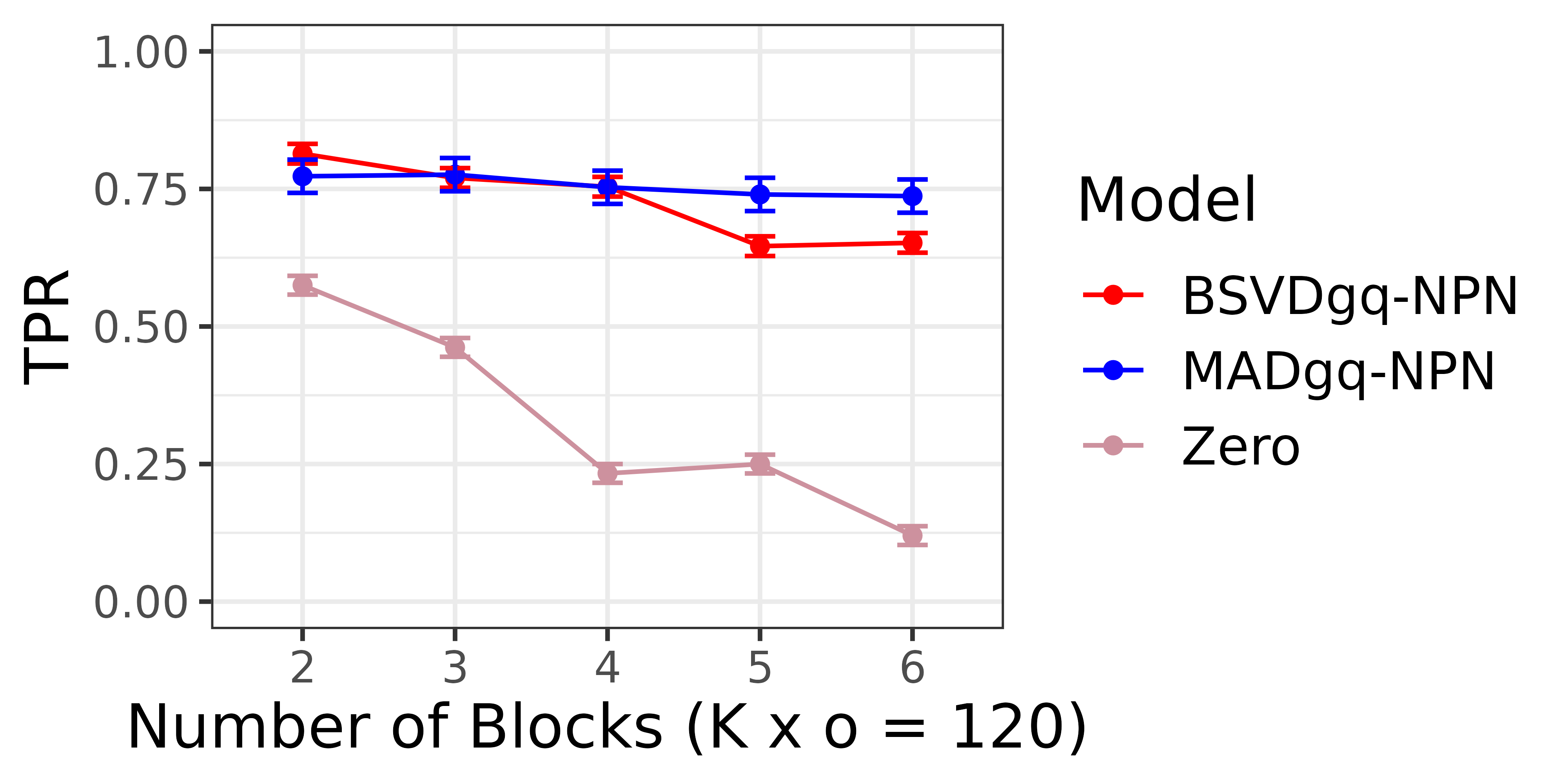}
        \caption{}
    \label{fig:cauchy_k_tpr}
    \end{subfigure}
    \begin{subfigure}[t]{0.4\linewidth}
    \centering
        \includegraphics[width=\linewidth]{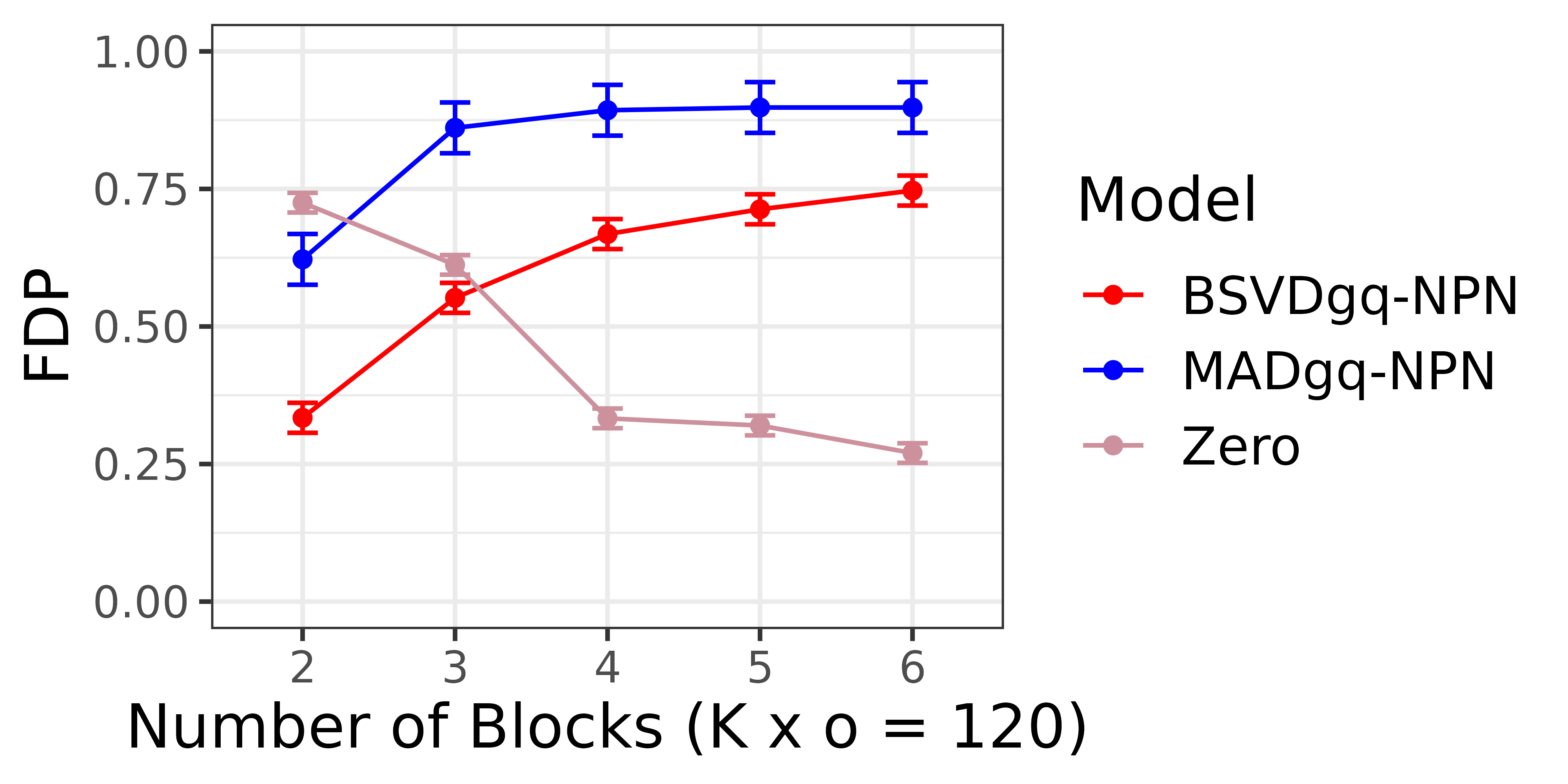}
        \caption{}
    \label{fig:cauchy_k_fdp}
    \end{subfigure}%
\end{center}
    \caption{Performance of MAD$_{\GQ}$-NPN, BSVDgq-NPN, and zero imputation on simulated data from a Cauchy distribution with low-rank covariance. \textbf{(a)} TPR, changing block size. \textbf{(b)} FDP, changing block size. \textbf{(c)} TPR, changing number of blocks. \textbf{(d)} FDP, changing number of blocks. }
    \label{fig:cauchy}
\end{figure}

\section{Calcium Imaging Example} \label{sec:cal}

We now study the nonparanormal Graph Quilting procedures on a real-world calcium imaging data set in order to assess the applicability of these methods for estimating functional neuronal connectivity. The data come from the Allen Institute \citep{abadata} and contains functional activity recordings for 227 neurons in a mouse V1 cortex during spontaneous activity across approximately 9000 time points. For this analysis, we compare the performances of the MAD$_{\GQ}$-NPN and BSVDgq-NPN nonparanormal Graph Quilting methods to each other, as well as to a zero imputation approach, in a similar fashion to the methodology in Section \ref{sec:sim}. We also compare the two nonparanormal Graph Quilting methods to their analogous Gaussian-based procedures. 

\begin{figure}[t]
\begin{center}
    \begin{subfigure}[t]{0.4\linewidth}
    \centering
        \includegraphics[width=\linewidth]{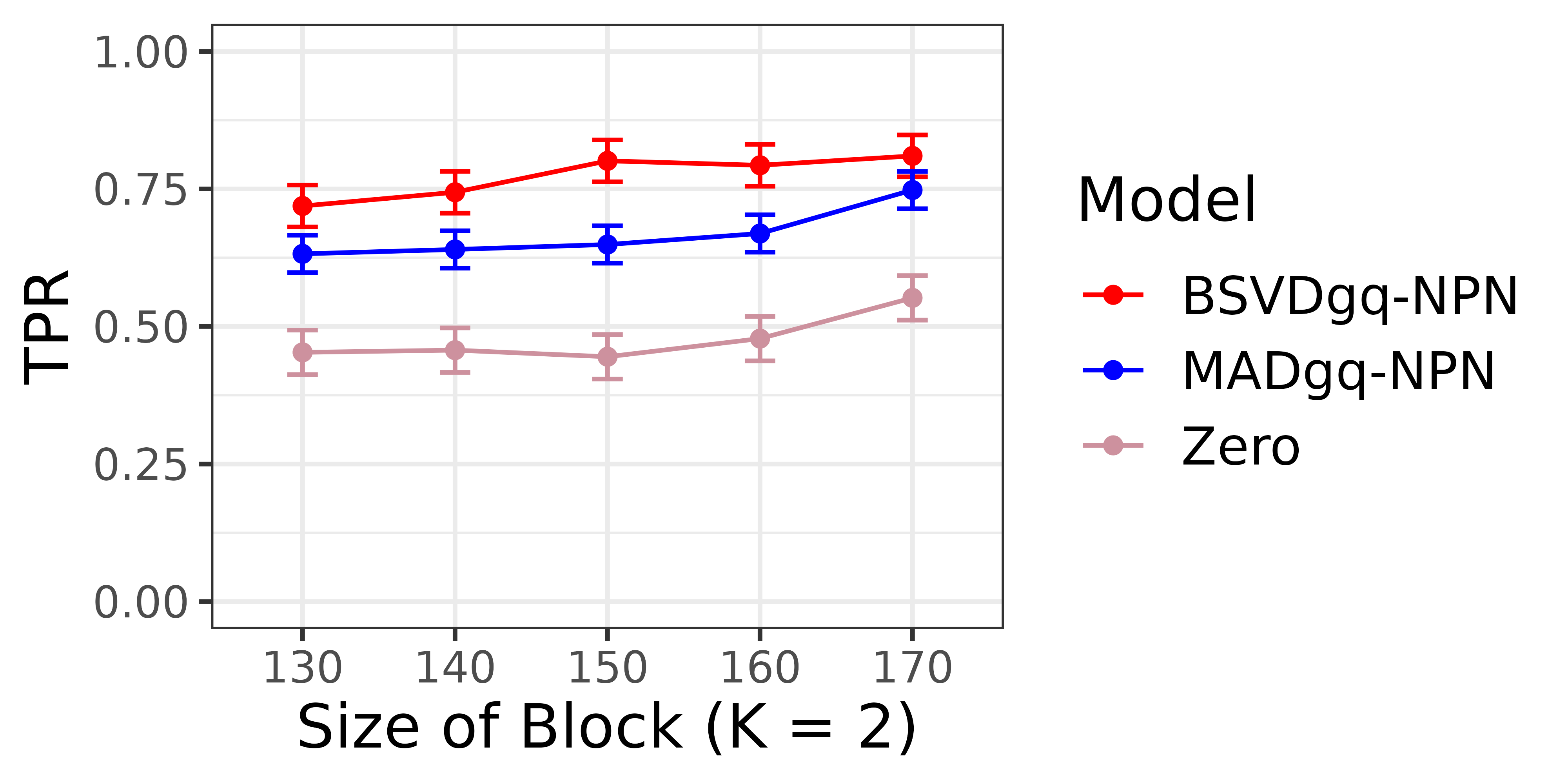}
        \caption{}
    \label{fig:aba_o_tpr}
    \end{subfigure}%
    \begin{subfigure}[t]{0.4\linewidth}
    \centering
        \includegraphics[width=\linewidth]{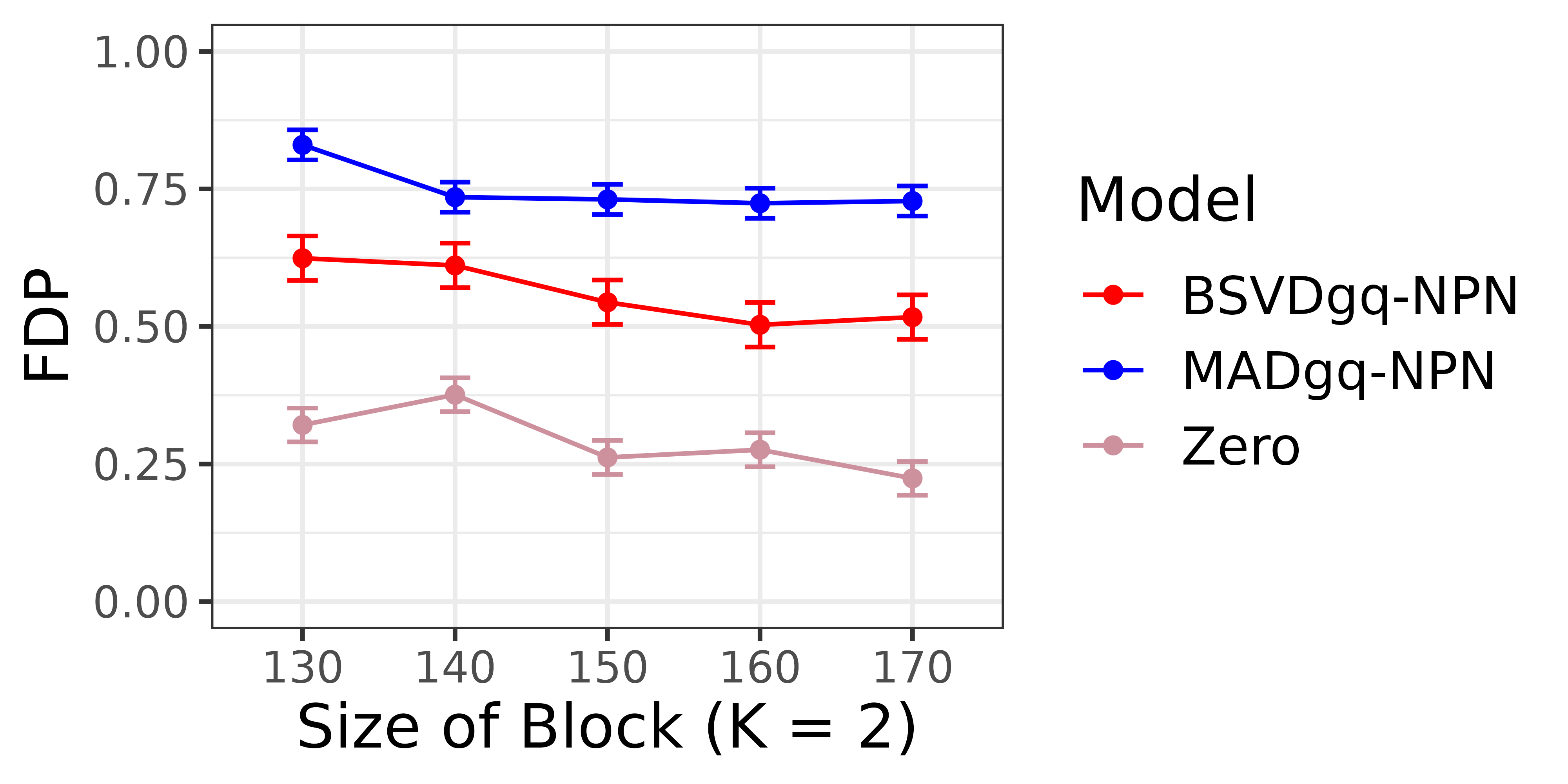}
        \caption{}
    \label{fig:aba_o_fdp}
    \end{subfigure}
    \par\medskip
    \begin{subfigure}[t]{0.4\linewidth}
    \centering
        \includegraphics[width=\linewidth]{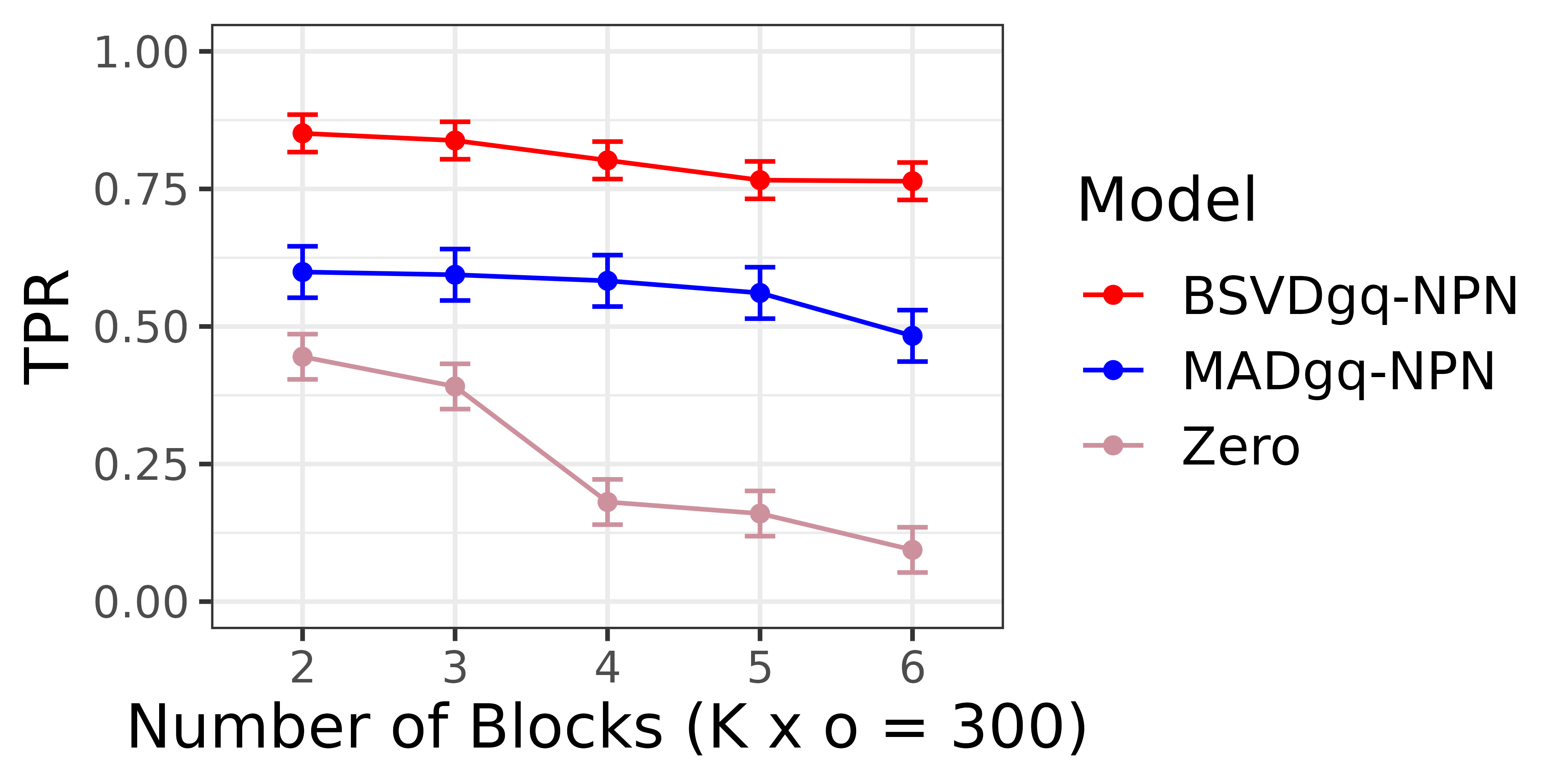}
        \caption{}
    \label{fig:aba_k_tpr}
    \end{subfigure}
    \begin{subfigure}[t]{0.4\linewidth}
    \centering
        \includegraphics[width=\linewidth]{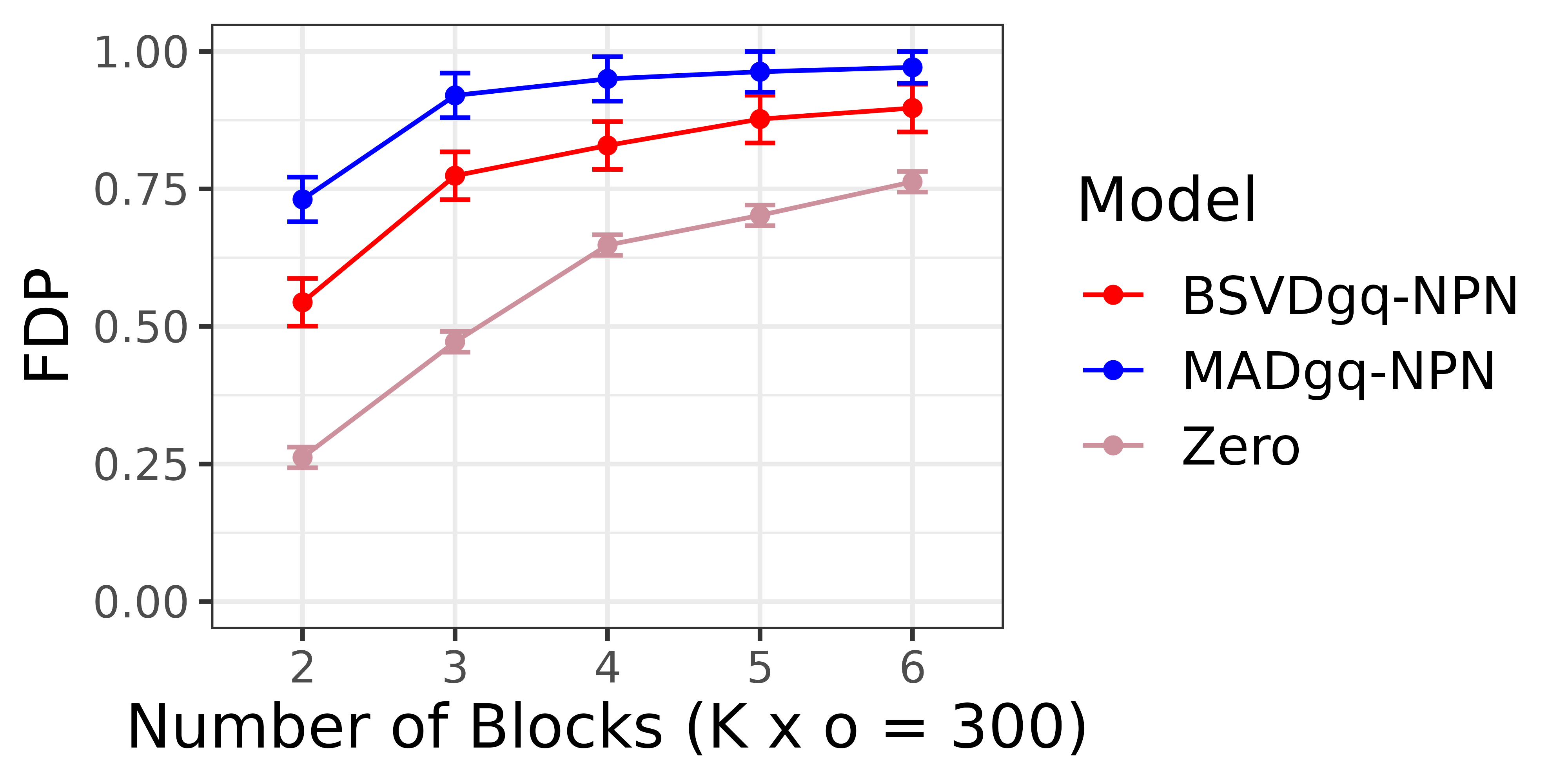}
        \caption{}
    \label{fig:aba_k_fdp}
    \end{subfigure}%
\end{center}
    \caption{Performance of MAD$_{\GQ}$-NPN, BSVDgq-NPN, and zero imputation on Allen Institute data. \textbf{(a)} TPR, changing block size. \textbf{(b)} FDP, changing block size. \textbf{(c)} TPR, changing number of blocks. \textbf{(d)} FDP, changing number of blocks.}
    \label{fig:aba}
\end{figure}

To do the former, we measure the performances of the MAD$_{\GQ}$-NPN, BSVDgq-NPN, and zero imputation procedure by how well the graph estimates from each nonparanormal Graph Quilting method on rank-based correlation matrices with synthetic block-missingness recovers the edge structure of the graph estimated using the nonparanormal graphical model on the same rank-based correlation matrix with all pairwise entries observed. Specifically, we calculate the Spearman correlation matrix for all observed neurons using binned spike counts and apply the graphical Lasso to the fully observed covariance, the result of which we consider as the true underlying graph structure. We create artificial block-missingness in the empirical covariance matrix by randomly assigning features to $K$ partially overlapping blocks of size $o$ and masking all pairwise entries in the covariance which are not contained in any block. The masked Spearman correlation matrix is then used as input for the nonparanormal Graph Quilting methods. We show the average and standard deviation of the TPR and FDR for recovering the graph estimate on the full Spearman covariance matrix for each method across 50 replications on each set of parameters, with new random block assignments each replication. Hyperparameter selection is performed using optimal F-1 score with respect to the graph estimated from the graphical Lasso fit on the fully observed data with the rank-based correlation matrix.

Figures \ref{fig:aba_o_tpr} and \ref{fig:aba_o_fdp} show the TPR and FDP of each method compared to the true underlying graph when the size of the block $o$ is 130, 140, 150, 160, and 170 while the number of blocks is held constant at 2, and Figures \ref{fig:aba_k_tpr} and \ref{fig:aba_k_fdp} show the TPR and FDP for $K = $ 2, 3, 4, 5 and 6 blocks while keeping the total number of observations across all blocks ($K \times o$) constant at 300. In general, both nonparanormal Graph Quilting methods are able to recover most of edges as the graph estimated with the nonparanormal graphical model when all features are observed simultaneously, which shows that both methods can reliably recover the edges in functional neuronal connectivity networks derived from calcium imaging data that would be present if all neurons were observed simultaneously. However, we also observe a relatively high FDP for both methods, which seems to show that both methods tend to slightly overselect the total number of edges in the underlying full graph. The BSVDgq-NPN method outperforms MAD$_{\GQ}$-NPN for edge recovery in terms of both TPR and FDP; this can likely be attributed to the approximate low-rank structure often found in the empirical covariance matrices from calcium imaging data \citep{stringerlr}. Also, as seen in Section \ref{sec:sim}, all methods perform better with fewer blocks and with larger block sizes.

\begin{figure}[t]
\begin{center}
    \begin{subfigure}[t]{0.4\linewidth}
    \centering
        \includegraphics[width=\linewidth]{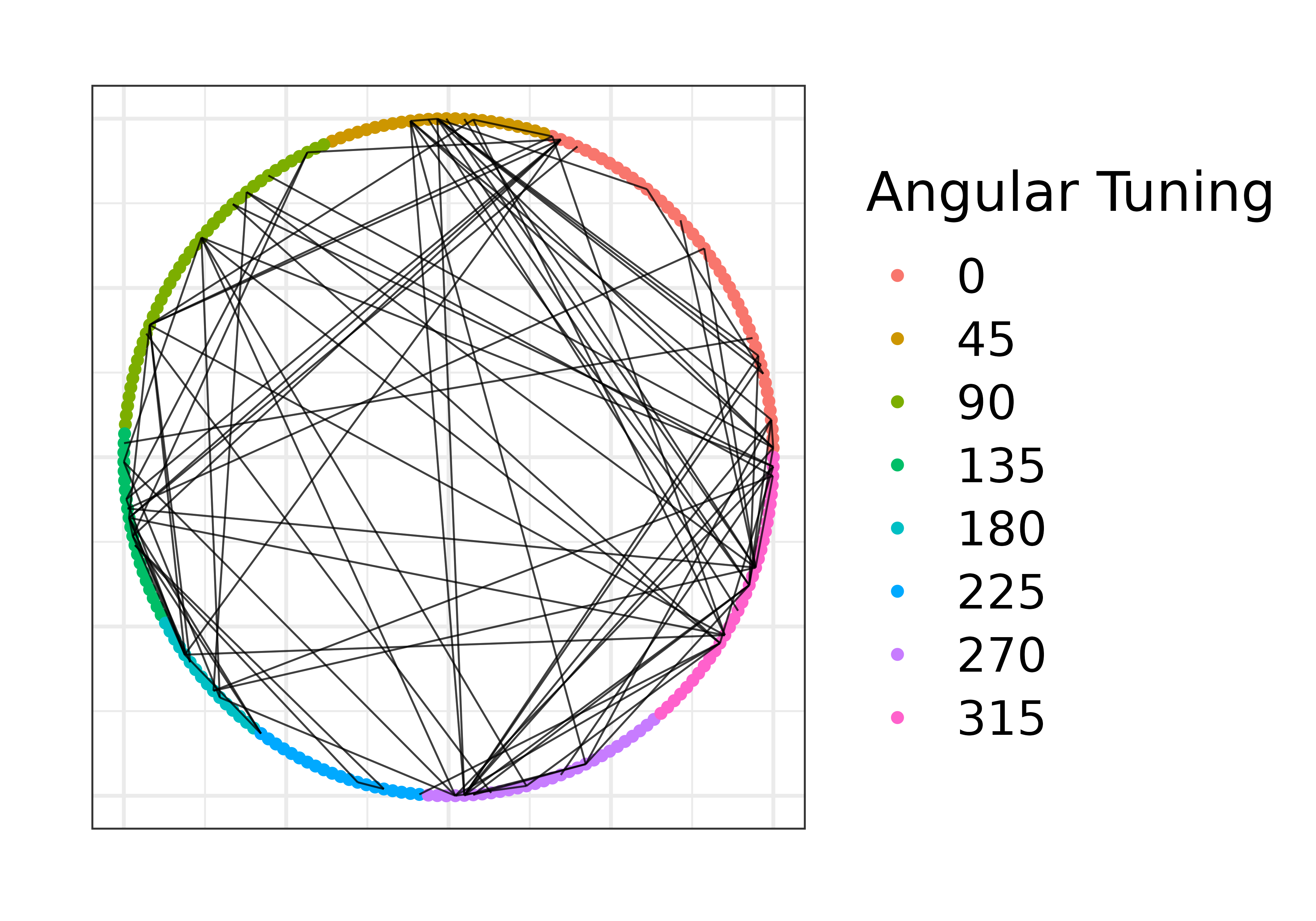}
        \caption{}
    \label{fig:aba_circ_gmm}
    \end{subfigure}%
    \begin{subfigure}[t]{0.4\linewidth}
    \centering
        \includegraphics[width=\linewidth]{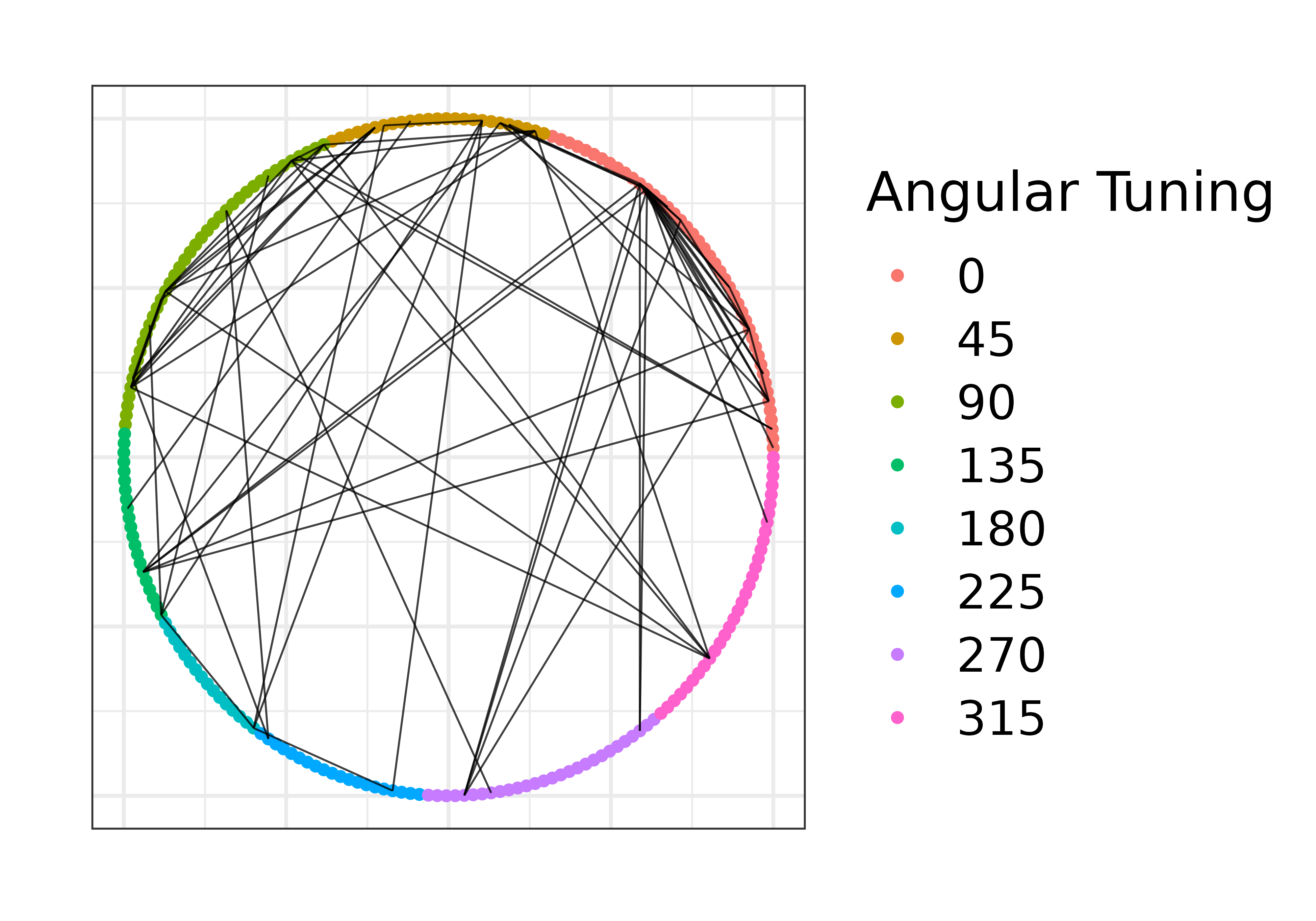}
        \caption{}
    \label{fig:aba_circ_npn}
    \end{subfigure}
    \begin{subfigure}[t]{0.4\linewidth}
    \centering
        \includegraphics[width=\linewidth]{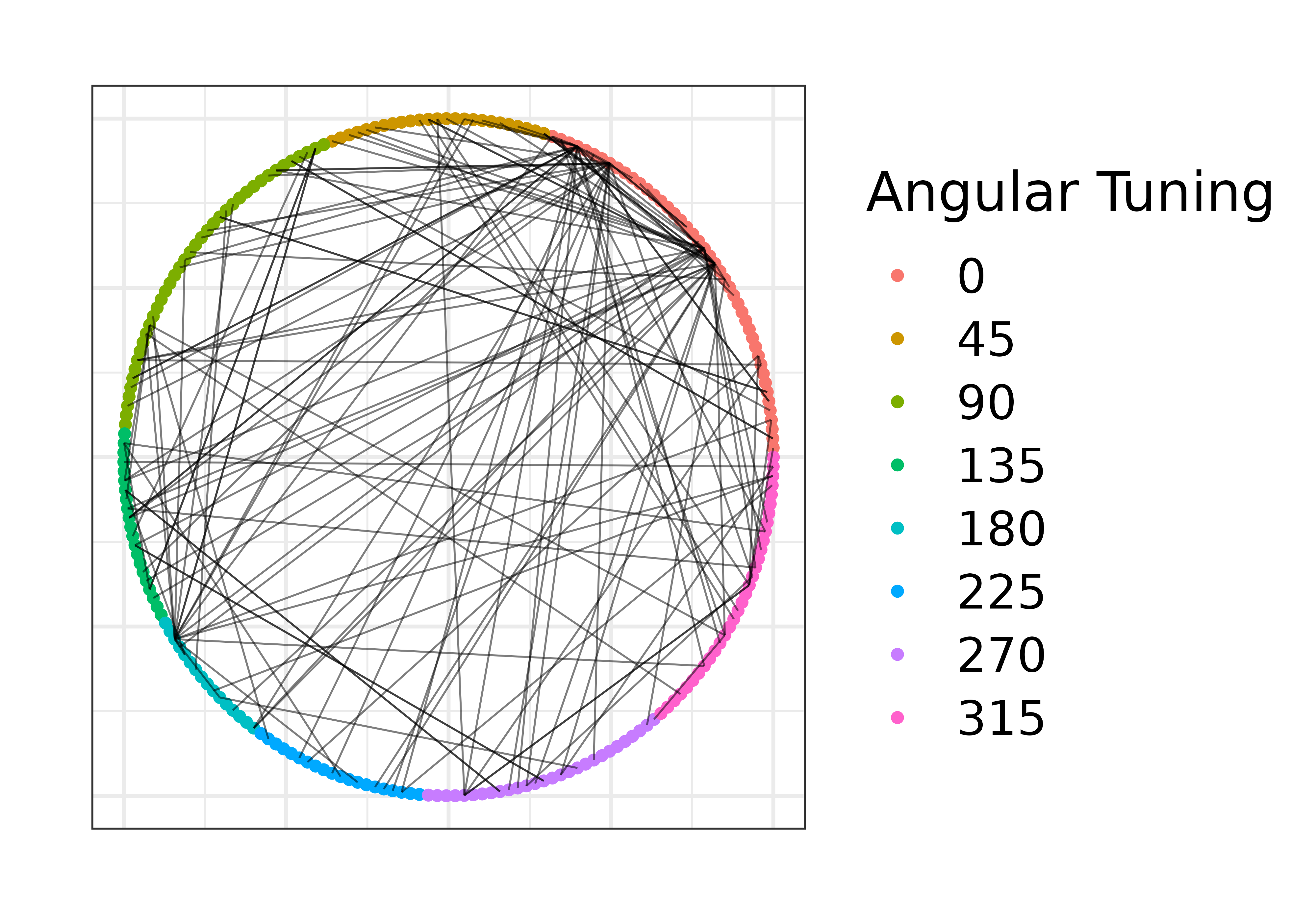}
        \caption{}
    \label{fig:aba_mad_circ_gmm}
    \end{subfigure}%
    \begin{subfigure}[t]{0.4\linewidth}
    \centering
        \includegraphics[width=\linewidth]{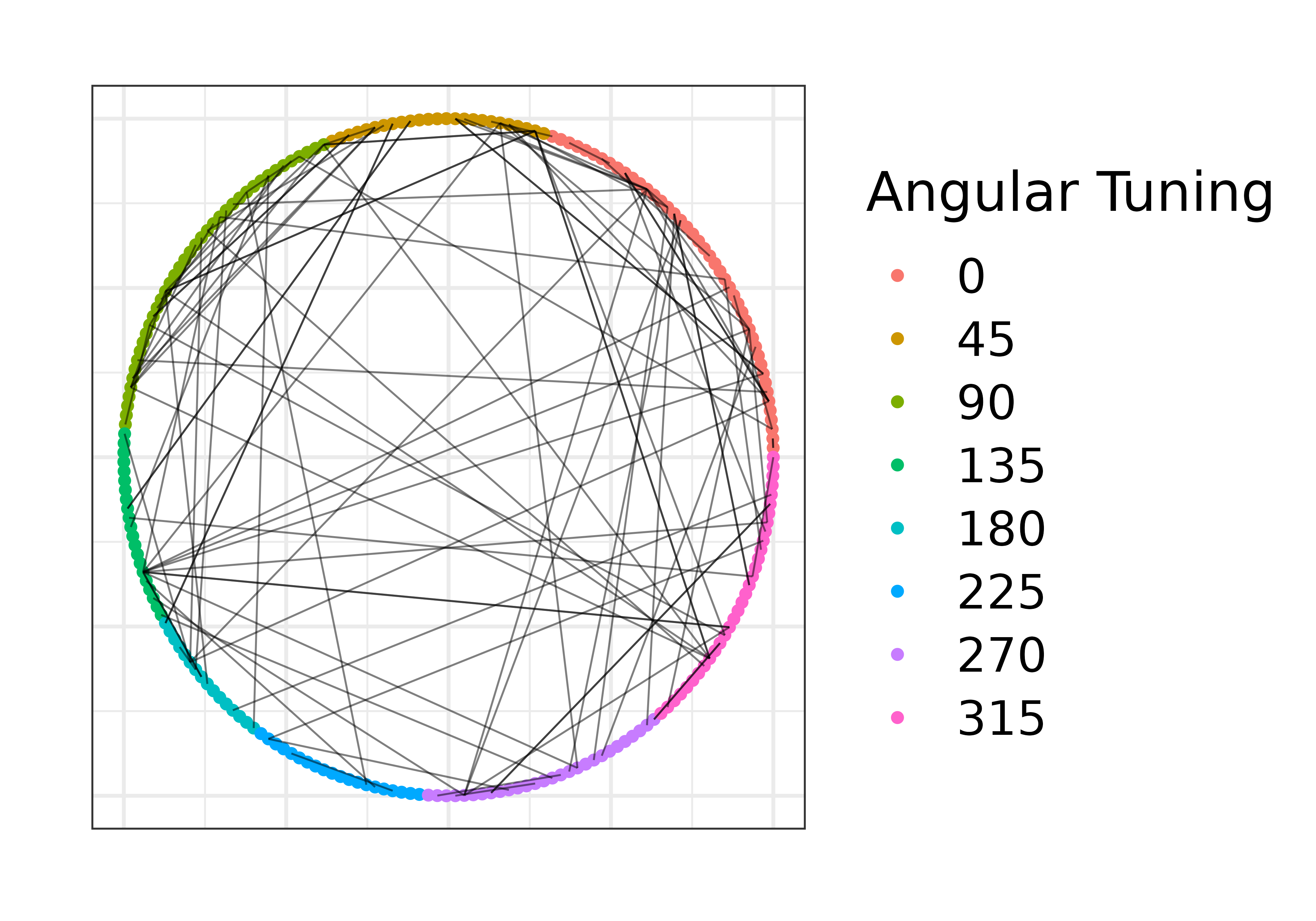}
        \caption{}
    \label{fig:aba_mad_circ_npn}
    \end{subfigure}
\end{center}
    \caption{Estimated functional connectivity networks from Graph Quilting methods on Allen Institute data. Nodes are colored by visual angular tuning category. \textbf{(a)} BSVDgq. \textbf{(b)} BSVDgq-NPN. \textbf{(c)} MAD$_{\GQ}$. \textbf{(d)} MAD$_{\GQ}$-NPN.}
    \label{fig:aba_circ}
\end{figure}

We then compare the functional neuronal connectivity graphs from the nonparanormal Graph Quilting methods to those found using Graph Quilting with a Gaussian assumption. Specifically, we use the MAD$_{\GQ}$ and BSVDgq Graph Quilting methods to obtain graph estimates on the same data set with the same artificial block-missingness pattern as above, and compare the selected edges to functional connectivity graphs estimated by MAD$_{\GQ}$-NPN and BSVDgq-NPN using neural properties from provided metadata. We first assess each of the estimated functional connectivity networks by the proportion of edges that connect pairs of neurons with the same visual angular tuning category. In the neuroscience literature, it has been hypothesized that neurons are tuned such that they fire in the presence of specific stimuli \citep{tune} and that neurons with the similar tunings are more likely to be functionally connected \citep{tune2}. Thus, we expect a sizable proportion of the edges in the estimated functional connectivity networks to link pairs of neurons within the same tuning category. For this particular calcium imaging data set, neural angular tuning is categorized into 8 different bins, each comprising 45 degree intervals. Hyperparameter tuning for the BSVDgq and BSVDgq-NPN methods is performed using the extended Bayesian information criterion with respect to the original data; for the MAD$_{\GQ}$ and MAD$_{\GQ}$-NPN methods, in order to crate a fair comparison between methods, we tune the thresholding parameters such that the number of edges is approximately similar to the estimates from BSVDgq and BSVDgq-NPN. 

Figures \ref{fig:aba_circ_gmm} and \ref{fig:aba_circ_npn} show examples of estimated functional connectivity networks from the BSVDgq and BSVDgq-NPN methods, and Figures \ref{fig:aba_mad_circ_gmm} and \ref{fig:aba_mad_circ_npn} show estimated functional connectivity networks from the MAD$_{\GQ}$ and MAD$_{\GQ}$-NPN methods, respectively. Structurally, we see that the graph estimates from the nonparanormal Graph Quilting methods are much more likely to be comprised of edges between neurons with the same angular tuning category. Specifically, across different replications with synthetic block-missingness, the BSVDgq-NPN graph estimates contain an average of 38.9\% of edges that link neurons with the same tuning category, compared to 21.4\% of edges in the BSVDgq estimates. Similarly, we see that 46.8\% of edges in the MAD$_{\GQ}$-NPN graph estimates are between pairs of neurons in the same tuning bin, as opposed to just 28.7\% in the MAD$_{\GQ}$ graph estimates.

We also compare the recorded firing activity of one particular example neuron and its selected edge neighbors in the the functional connectivity graphs estimated from the BSVDgq and BSVDgq-NPN methods in terms of how closely the neural firing patterns match one another. Specifically, for the problem of functional neuronal connectivity, our goal is to find neurons with consistent synchronous firing activity across time, which is represented by contemporaneous large positive spikes in the fluorescence traces \citep{smetters1996, turaga2013}. In Figures \ref{fig:aba_ggm_1} and \ref{fig:aba_ggm_2}, we visualize the fluorescence trace of the selected neuron, overlayed with the fluorescence traces neurons which are edge neighbors unique to the BSVDgq graph; we also do the same in Figures \ref{fig:aba_nn_1} and \ref{fig:aba_nn_2} for neurons that are edge neighbors unique to the BSVDgq-NPN graph. The top 10 periods of spiking activity of the example selected neuron are represented via blue dotted lines in the plots. From the results, we see that the firing activity of the selected neuron seems to match relatively closely with the edge neighbors selected only in the graph estimate from the BSVDgq-NPN algorithm. On the other hand, the edge neighbors selected only by the BSVDgq graph do not appear to have the same firing pattern. Quantitatively, the top 10 firing times for the selected neuron and its edge neighbors match 72.4\% of the time for the BSVDgq-NPN functional connectivity graph estimates, as opposed to just 24.5\% for the BSVDgq functional connectivity graphs. Overall, from this real-world calcium imaging data study, we see that nonparanormal Graph Quilting provides more logical functional connectivity estimates in the nueroscience context compared to the ordinary Graph Quilting procedures.

\begin{figure}[t]
\begin{center}
    \begin{subfigure}[t]{0.4\linewidth}
    \centering
        \includegraphics[width=\linewidth]{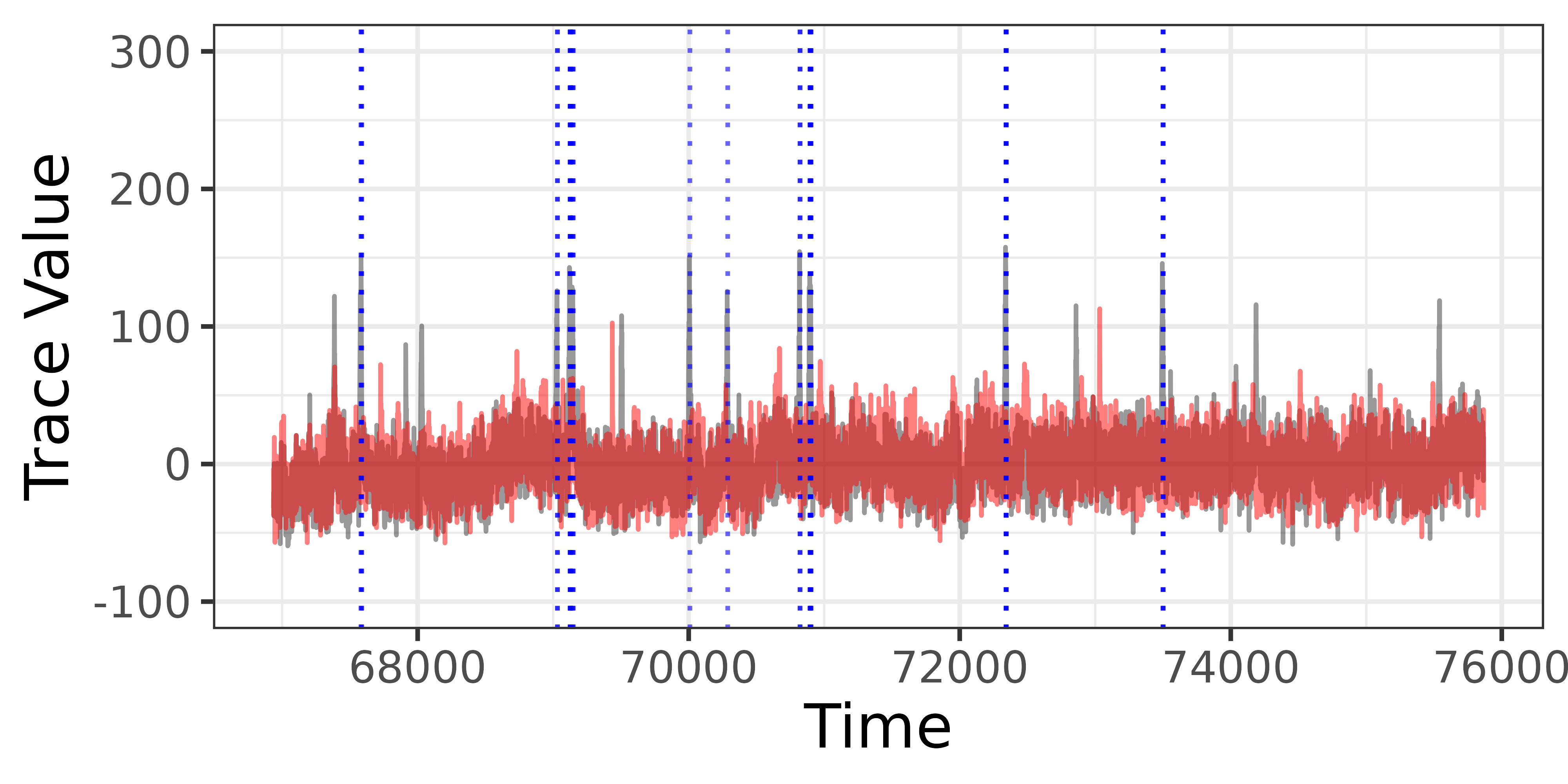}
        \caption{}
    \label{fig:aba_ggm_1}
    \end{subfigure}%
    \begin{subfigure}[t]{0.4\linewidth}
    \centering
        \includegraphics[width=\linewidth]{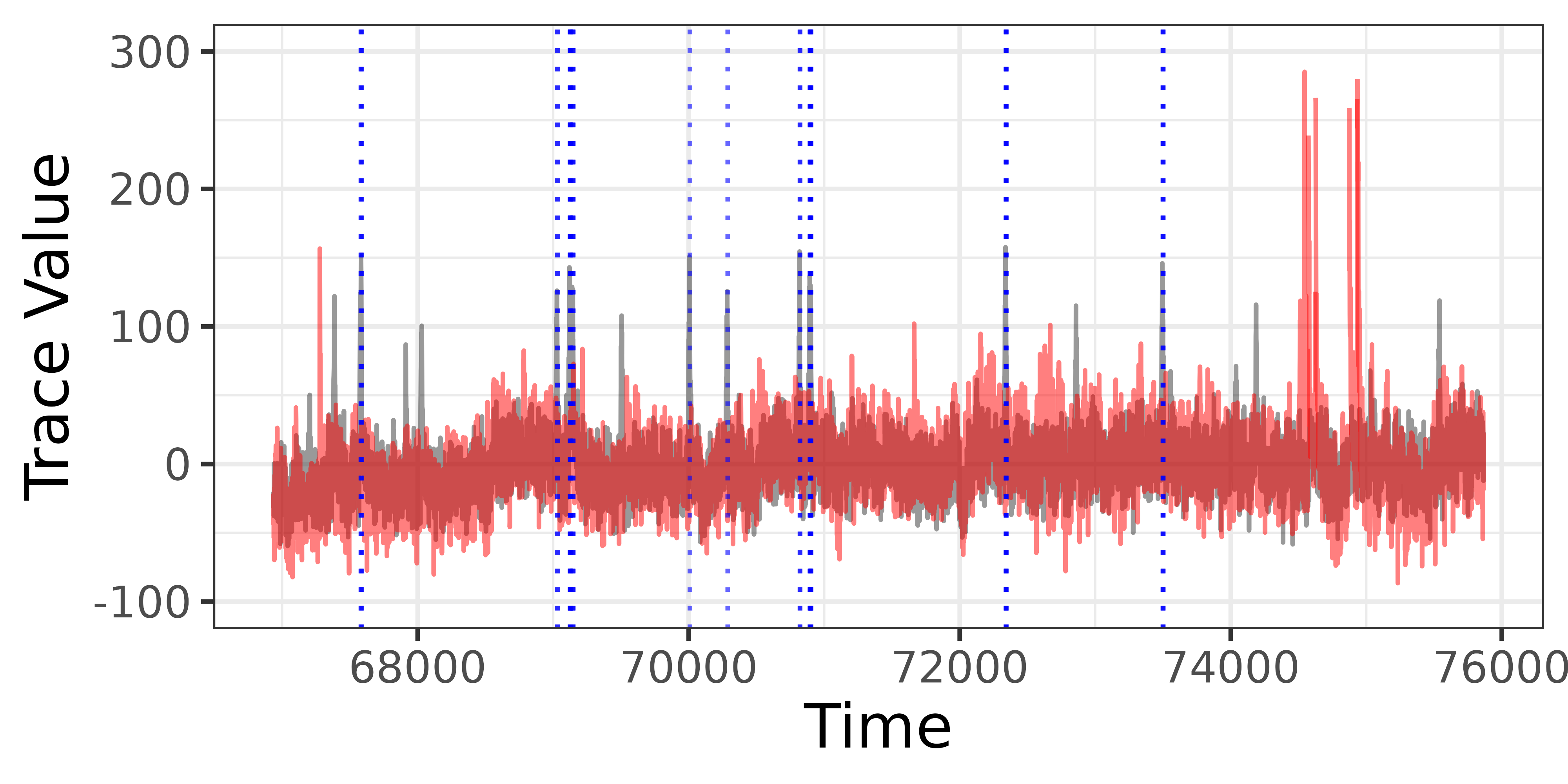}
        \caption{}
    \label{fig:aba_ggm_2}
    \end{subfigure}
    \par\medskip
    \begin{subfigure}[t]{0.4\linewidth}
    \centering
        \includegraphics[width=\linewidth]{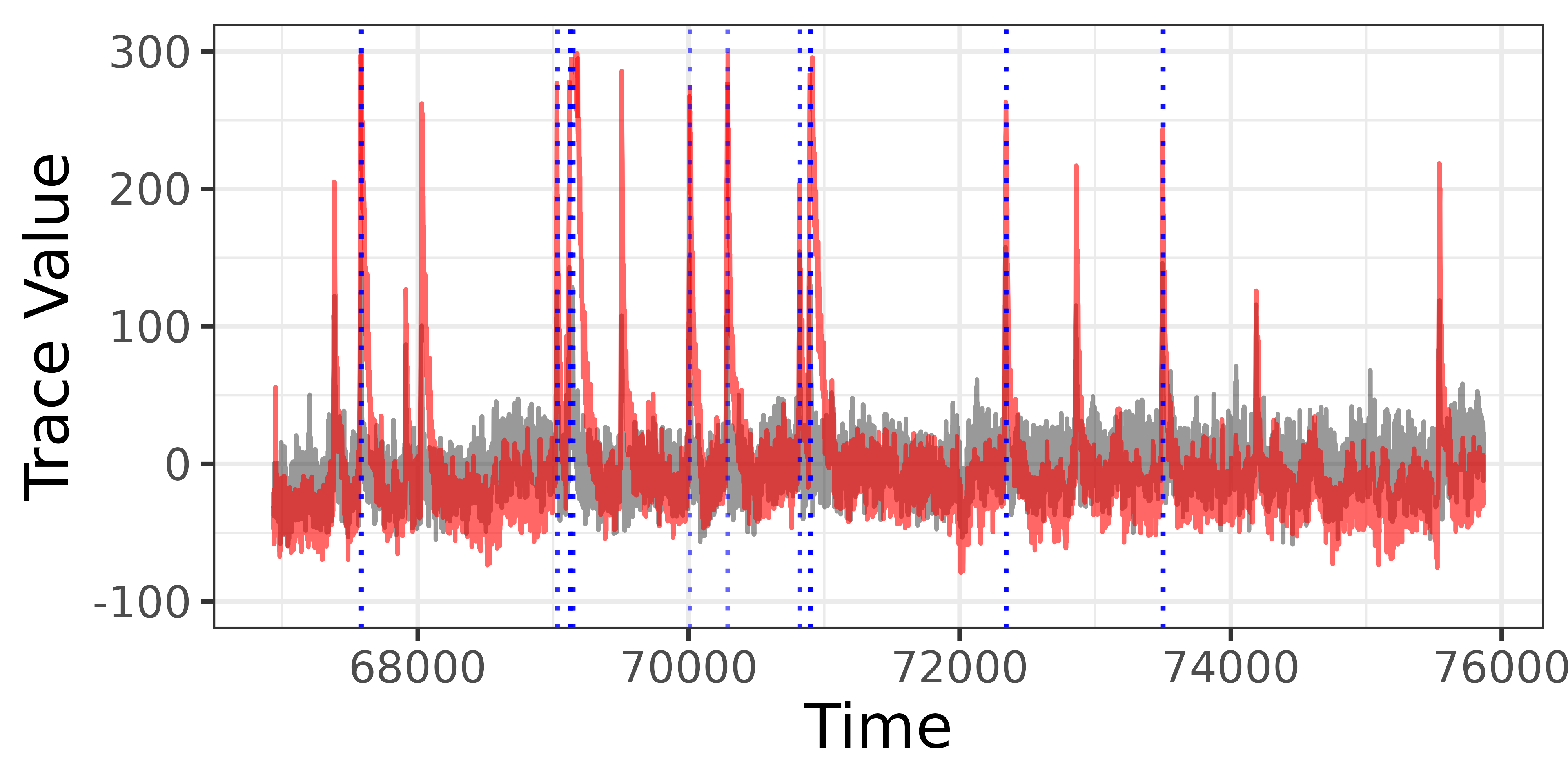}
        \caption{}
    \label{fig:aba_nn_1}
    \end{subfigure}
    \begin{subfigure}[t]{0.4\linewidth}
    \centering
        \includegraphics[width=\linewidth]{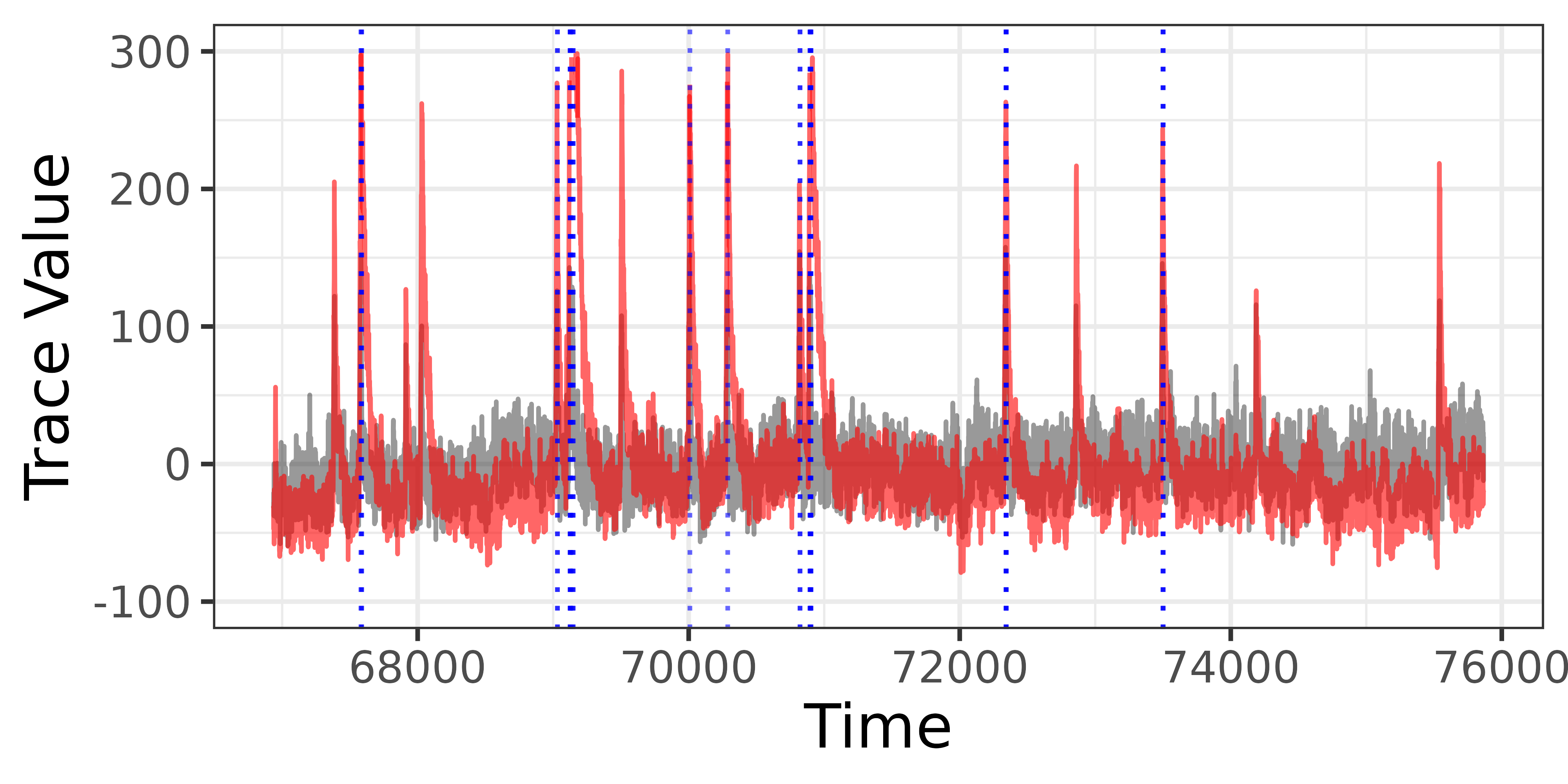}
        \caption{}
    \label{fig:aba_nn_2}
    \end{subfigure}%
\end{center}
    \caption{Fluorescence traces of one particular neuron (in grey) and edge neighbors (in red) from functional connectivity graphs estimated via BSVDgq (\textbf{(a, b)}) and BSVDgq-NPN (\textbf{(c, d)}).}
    \label{fig:flaba}
\end{figure}

\section{Discussion}

In this work, we have presented two potential approaches to nonparanormal Graph Quilting, MAD$_{\GQ}$-NPN and DSVDgq-NPN, which broaden the scope of Graph Quilting procedures to be applicable in the nonparanormal graphical model setting. We demonstrate theoretical properties of the MAD$_{\GQ}$-NPN method, showing criterion for exact edge recovery in the observed portion and minimum superset recovery in the missing portion of the graph. Through our empirical studies, we demonstrate that both nonparanormal Graph Quilting methods can be effective for edge selection for non-Gaussian data depending on the structure of the underlying covariance matrix. Through our real-world calcium imaging data example, we show that the nonparanormal Graph Quilting methods can be used to recover the same functional neuronal connectivity network edges from calcium imaging data in the presence of non-simultaneous observations for the full population of neurons as would be found if all neurons are observed concurrently, and that these methods can be applied to estimate more appropriate functional neuronal connectivity networks from calcium imaging data compared to Gaussian methods.

There are many potential directions for future research that can be taken from our work. While we have characterized the theoretical performance of the MAD$_{\GQ}$-NPN approach, we do not currently have any guarantees for the LRGQ-NPN methods. Methodologically, extensions to the nonparanormal Graph Quilting procedure to account for potential other data effects such as latent variables, autocorrelations, or covariates could improve graph estimation and edge selection accuracy. Also, the nonparanormal graphical models could be applied to research problems in other fields where joint observations may be missing, such as RNA-seq in genomics and signal processing in power systems. In conclusion, our work has helped to extend graph inference for nonparanormal graphical models in the presence of block-missingness in the observed covariance matrix, with theoretical guarantees for performance and promising empirical results for calcium imaging data.

\section*{Acknowledgements}

The authors gratefully acknowledge support by NSF NeuroNex-1707400, NIH 1R01GM140468, and NSF DMS-2210837.

\bibliographystyle{apalike}
\bibliography{reference}

\newpage

\appendix
\setcounter{theorem}{0}
\setcounter{assump}{0}

\section{Algorithmic Details}
In our LRGQ-NPN approach, one key step is to impute the covariance matrix $\h{\Sigma}^{\LR}$ from partially computed $\h{\Sigma}_O = \h{\Sigma}_O^{(\tau)}$ or $\h{\Sigma}_O^{(\rho)}$. We consider the following BSVDgq algorithm (Algorithm \ref{alg:sbsvdgq}) proposed in \citep{chang2022low} for this step.
\begin{algorithm}
  \caption{Spiked Block Singular Value Decomposition (Spiked BSVD\textsubscript{gq}) \citep{chang2022low}} \label{alg:sbsvdgq}
   \textbf{Input:} $\{V_k, k \in 1, \hdots K\}$, $\widehat{\Sigma}_O \in \mathbb{R}^{p\times p}$, $r > 0$.\\
  \textbf{Initialize:} $\widetilde{H} = \boldsymbol{0}_{p \times r}$, $A = V_1$, $\widehat{q} = \text{median}(\{\widehat{\Sigma}_{ii}, 1 \leq i \leq p\})$.\\
    Find low-rank solution for first patch: 
    \vspace{-2mm}
    \begin{enumerate}
        \item Calculate SVD of $(\widehat{\Sigma} - \widehat{q}\boldsymbol{I})_{V_1,\, V_1} = \mathbf{W}\boldsymbol{\Lambda}\mathbf{U}^T.$ 
        \vspace{-2mm}
        \item Set $\mathbf{h}$ as indices of the largest $r$ diagonal elements of $\boldsymbol{\Lambda}$.
        \vspace{-2mm}
        \item Set $\mathbf{C}_{V_1,\, :} = \mathbf{U}_{:, \mathbf{h}} \boldsymbol{\Lambda}_{\mathbf{h}, \mathbf{h}}^{1/2}.$
    \end{enumerate}
    \vspace{-2mm}
    For $s \in 2, \hdots K$
    \vspace{-2mm}
   \begin{enumerate}
       \item Find low-rank solution for $s$-th patch: 
       \vspace{-2mm}
       \begin{enumerate}
           \item Calculate SVD of $(\widehat{\Sigma} - \widehat{q}\boldsymbol{I})_{V_s,\, V_s} = \mathbf{W}\boldsymbol{\Lambda}\mathbf{U}^T.$
           \vspace{-2mm}
           \item Set $\mathbf{h}$ as indices of the largest $r$ diagonal elements of $\boldsymbol{\Lambda}$.
           \vspace{-2mm}
           \item Calculate $\mathbf{D} = \mathbf{U}_{:, \mathbf{h}} \boldsymbol{\Lambda}_{\mathbf{h}, \mathbf{h}}^{1/2}.$
       \end{enumerate}
       \vspace{-5mm}
       \item Merge with previous patches:
       \vspace{-2mm}
       \begin{enumerate}
           \item Find overlaps $E = \{A: a \in V_s\}, \, J = \{j: V_s[j] \in A\}$.
           \vspace{-2mm}
           \item Calculate SVD of $(\mathbf{D}_{J}^T \mathbf{C}_{E,\, :}) = \mathbf{W}\boldsymbol{\Lambda}\mathbf{U}^T.$
           \vspace{-2mm}
           \item Set $\mathbf{M} = \mathbf{C}_{E,\, :}, \mathbf{C}_{V_s,\, :} = \mathbf{D}\mathbf{W}\mathbf{U}^T$
           \vspace{-2mm}
           \item Set $\mathbf{C}_{E,\, :} = \mathbf{M}$\\
       \end{enumerate}
       \vspace{-8mm}
       \item Update $A = \bigcup_{k = 1}^s V_s $
   \end{enumerate}
   \vspace{-2mm}
    \Return{$\widetilde{\Sigma} = \mathbf{C}\mathbf{C}^\top.$}
\end{algorithm}

\section{Proofs}
For completeness, here we restate our MAD$_{\GQ}$-NPN algorithm, its theory and assumptions presented in the main paper.
The rank-based correlation is defined as follows:
\begin{equation}\label{eq:app_Sigma_rank}
    \widehat{\Sigma}^{(\rho)}_{j,l} = 2\sin\left(\frac{\pi}{6}\widehat{\rho}_{j,l}\right),\quad \widehat{\Sigma}^{(\tau)}_{j,l} = \sin\left(\frac{\pi}{2}\widehat{\tau}_{j,l}\right),
\end{equation}
where
\begin{equation}\label{eq:app_rho_tau}
    \widehat{\rho}_{j,l} = \frac{\sum_{k=1}^K\ind{j,l\in V_k}\widehat{\rho}^{(k)}_{j_k,l_k}}{\sum_{k=1}^K\ind{j,l\in V_k}}, \quad \widehat{\tau}_{j,l} = \frac{\sum_{k=1}^K\ind{j,l\in V_k}\widehat{\tau}^{(k)}_{j_k,l_k}}{\sum_{k=1}^K\ind{j,l\in V_k}},
\end{equation}
and 
\begin{equation}\label{eq:app_rho_tau_block}
\begin{split}
    \widehat{\rho}^{(k)}_{j,l} &= \frac{\sum_{i=1}^{n_k}(r^{(k)}_{i,j} - \bar{r}^{(k)}_{j})(r^{(k)}_{i,l} - \bar{r}^{(k)}_{l})}{\sqrt{\sum_{i=1}^{n_k}(r^{(k)}_{i,j} - \bar{r}^{(k)}_{j})^2\sum_{i=1}^{n_k}(r^{(k)}_{i,l} - \bar{r}^{(k)}_{l})^2}},\\
    \widehat{\tau}^{(k)}_{j,l} &= \frac{2}{n_k(n_k-1)}\sum_{1\leq i<i'\leq n_k}\mathrm{sign}((r^{(k)}_{i,j} - r^{(k)}_{i',j})(r^{(k)}_{i,l} - r^{(k)}_{i',l})).
\end{split}
\end{equation}

\begin{assump}[Weak distortion compared to signal]\label{assump:app_distortion_signal}
    We assume that the maximum off-diagonal distortion of the MAD$_{\GQ}$ solution is smaller than half the signal strength in the original precision matrix: $\delta<\frac{\nu}{2}$.
\end{assump}

\begin{assump}\label{assump:app_Schur_distortion1}
     For every node $i\in V$ with $N_{H_i}(i)\neq\emptyset$, we have that for every $k$ such that $i\in V_k$, there exists at least one node $j\in V_k\setminus \{i\}$ that is $(H_i\cup\{j\})$-connected to some node in $N_{H_i}(i)$.
\end{assump}

\begin{assump}\label{assump:app_Schur_distortion2}
    If $\delta_{i,\backslash i}^{(k)}\neq 0$, then there exists $j\neq i$ such that $0<|\tilde\Theta_{ij}^{(k)}|<\delta$.
\end{assump}

\begin{assump}[Incoherence condition]\label{assump:app_incoh}
    Let $\Gamma = \widetilde{\Sigma}\otimes \widetilde{\Sigma}$, $S = \{(j,l): \widetilde{\Theta}_{j,l}\neq 0\}$. We assume $\max_{e\in O\cap S^c}\|\Gamma_{e,S}\Gamma_{S,S}^{-1}\|_1\leq 1-\alpha$ for some $0<\alpha\leq 1$.
\end{assump}

\begin{assump}[Sufficient block measurements]\label{assump:app_V}
    The $K$ blocks cover all nodes: $\cup_{k=1}^KV_k = [p]$, and at least one off-diagonal element: $|O|>p$. 
\end{assump}

\begin{assump}[Regularization parameter]\label{assump:app_lambda}
    $\Lambda_{j,l} = \frac{C_0}{\alpha}\sqrt{\frac{\log p}{\min_k n_k}}$ for all $(j,l)\in O$ and some universal constant $C_0>0$.
\end{assump}

\begin{theorem}\label{thm:app_main}
Suppose that Assumptions \ref{assump:app_distortion_signal}-\ref{assump:app_lambda} hold, and there exist at least one edge in the graph encoded by $\Theta$ and $\til{\Theta}$: $d,\til{d}>2$. 
Then we have the following guarantees for Algorithm \ref{alg:MADgq}, with probability at least $1-\sum_kp_k^{-10}$:
\begin{itemize}
    \item \emph{\textbf{Exact recovery in $O$}}. If
\begin{equation}\label{eq:app_n_cond1}
    n_k> \left[\frac{C_0}{4}\kappa_{\til{\Gamma}}\left(1+\frac{8}{\alpha}\right)\left(\left(\frac{\nu}{2}-\delta\right)^{-1}+3\left(1+\frac{8}{\alpha}\right)(\kappa_{\til{\Sigma}}+\kappa^3_{\til{\Sigma}}\kappa_{\til{\Gamma}})\til{d}\right)\right]^2\log p_k,
\end{equation}
 $\delta+\varepsilon_1\leq \tau_1<\nu-\delta-\varepsilon_1$, where $\varepsilon_1 = \frac{C_0}{4}\kappa_{\til{\Gamma}}(1+\frac{8}{\alpha})\max_k\sqrt{\frac{\log p_k}{n_k}}$, then $\h{E}_O = E_O$.
    \item \emph{\textbf{Minimal superset recovery in $O^c$}}. If
\begin{equation}\label{eq:app_n_cond2}
n_k> C_0\kappa_{\til{\Gamma}}^2\left(1+\frac{8}{\alpha}\right)^2\left[\frac{9\til{\kappa}^4}{4\psi^2}+\frac{1}{4\lambda_{\min}(\til{\Theta})^2}\right]\min\{p+\til{s},\til{d}^2\}\log p_k,
\end{equation}
$\varepsilon_2\leq \tau_2< \psi - \varepsilon_2$, $\delta-\varepsilon_2< \tau_1\leq \nu-\varepsilon_2$, where $\varepsilon_2= \frac{3C_0}{4}\kappa_{\til{\Gamma}}(1+\frac{8}{\alpha})\til{\kappa}^2\min\{\sqrt{p+\til{s}},\til{d}\}\max_k\sqrt{\frac{\log p_k}{n_k}}$, then $\h{E}_{O^c} = \cS_{\off}$.
\end{itemize}
\end{theorem}
\begin{proof}[Proof of Theorem \ref{thm:app_main}]
First of all, we note that a population version of Theorem \ref{thm:app_main} has been proved in \citep{vinci2019graph} (see Theorem 3.1 and Theorem 3.4), which shows the exact edge recovery in $O$ and minimal superset recovery in $O^c$ based on thresholding $\til{\Theta}$ and $\til{\Theta}^{(k)}$. Our proof of Theorem \ref{thm:app_main} mainly involves applying the population results in \citep{vinci2019graph} and some new estimation error bounds for $\|\h{\til{\Theta}}-\til{\Theta}\|_{\infty}$ and $\|\h{\til{\Theta}}^{(k)}-\til{\Theta}^{(k)}\|_{\infty}$.

In particular, when Assumptions \ref{assump:app_distortion_signal}-\ref{assump:app_Schur_distortion2} hold, \citep{vinci2019graph} shows that 
\begin{equation}\label{eq:app_popEO}
    E_O = \{(i,j)\in O: |\til{\Theta}_{i,j}|> t\}, \quad \forall \delta\leq t\leq \nu-\delta,
\end{equation}
and for any $\delta\leq t\leq \nu$,
\begin{equation}\label{eq:app_popEOc}
    \cS_{\off} = O^c\cap (W_t\times W_t), W_t = \big\{i\in V:~ \forall k ~s.t.~i\in V_k,~\exists j\neq i, 0<|\tilde\Theta^{(k)}_{ij}|< t\big\}.
\end{equation}
On the other hand, the following two theorems show the proximity between $\h{\til{\Theta}}$ and $\til{\Theta}$, $\h{\til{\Theta}}^{(k)}$ and $\til{\Theta}^{(k)}$:
\begin{theorem}\label{thm:app_Theta_tilde_est_err}
Suppose that Assumptions \ref{assump:app_incoh} and \ref{assump:app_V} hold, $\Lambda_{j,l} = \frac{C_0}{\alpha}\sqrt{\frac{\log p}{\min_k n_k}}$ for all $(j,l)\in O$ and some universal constant $C_0>0$, and there exist at least one edge in the graph encoded by $\Theta$ and $\til{\Theta}$: $d,\til{d}>2$. If for all $1\leq k\leq K$, $$n_k\geq \left[\frac{3C_0}{4}\left(1+\frac{8}{\alpha}\right)^2(\kappa_{\widetilde{\Sigma}}\kappa_{\til{\Gamma}}+ \kappa_{\widetilde{\Sigma}}^3\kappa_{\widetilde{\Gamma}}^2)\right]^2\widetilde{d}^2\log p_k,$$ then with probability at least $1-\sum_{k=1}^Kp_k^{-10}$, $\widehat{\widetilde{\Theta}}$ defined in \eqref{eq:MADgqLasso} satisfies
$$\|\widehat{\widetilde{\Theta}}-\widetilde{\Theta}\|_{\infty}\leq \frac{C_0}{4}\kappa_{\til{\Gamma}}(1+\frac{8}{\alpha})\max_k\sqrt{\frac{\log p_k}{n_k}}.$$
\end{theorem}
\begin{theorem}\label{thm:app_Theta_Schur_err}
    Suppose all conditions in Theorem \ref{thm:app_Theta_tilde_est_err} hold. In addition, if for $1\leq k\leq K$,
    $$
    n_k\geq \frac{C_0\kappa_{\til{\Gamma}}^2(1+\frac{8}{\alpha})^2}{4\lambda_{\min}(\til{\Theta})^2}\min\{p+\til{s},\til{d}^2\}\log p_k,
    $$
    then for any $1\leq k\leq K$,
    $$\|\h{\til{\Theta}}^{(k)} - \til{\Theta}^{(k)}\|_{\infty}\leq \frac{3C_0}{4}\til{\kappa}^2\kappa_{\til{\Gamma}}(1+\frac{8}{\alpha})\min\{\sqrt{p+\til{s}},\til{d}\}\max_k\sqrt{\frac{\log p_k}{n_k}}$$
    holds with probability at least $1-\sum_{k=1}^Kp_k^{-10}$.
\end{theorem}
When the sample size condition \eqref{eq:app_n_cond1} holds, the conditions in Theorem \ref{thm:app_Theta_tilde_est_err} are all satisfied, and hence $\|\widehat{\widetilde{\Theta}}-\widetilde{\Theta}\|_{\infty}\leq \varepsilon_1$ defined in Theorem \ref{thm:app_main}. Thus when $\delta+\varepsilon_1\leq \tau_1\leq \nu-\delta-\varepsilon_1$, $$\h{E}_O = \{(i,j)\in O: i\neq j,|\h{\widetilde{\Theta}}_{i,j}|>\tau_1\}\subset \{(i,j)\in O: i\neq j,|\widetilde{\Theta}_{i,j}|>\delta\} = E_O,$$ and $$\h{E}_O \supset \{(i,j)\in O: i\neq j,|\widetilde{\Theta}_{i,j}|>\tau_1+\varepsilon_1\}\supset \{(i,j)\in O: i\neq j,|\widetilde{\Theta}_{i,j}|>\nu-\delta\} = E_O.$$ The sample size condition \eqref{eq:app_n_cond1} also ensures that $\nu-\delta-\varepsilon_1> \delta+\varepsilon_1$ so that the required range for $\tau_1$ is not empty.

In addition, when the sample size condition \eqref{eq:app_n_cond2} holds, all conditions in Theorem \ref{thm:app_Theta_Schur_err} are satisfied, and the range for $\tau_2$ is not empty. Thus $\|\widehat{\widetilde{\Theta}}^{(k)}-\widetilde{\Theta}^{(k)}\|_{\infty}\leq \varepsilon_2$. When $\varepsilon_2\leq \tau_2<\psi-\varepsilon_2$, $\delta-\varepsilon_2<\tau_3\leq \nu-\varepsilon_2$, the set $\h{W}_{\tau_2,\tau_3}$ defined in \eqref{eq:superset_node} satisfies
\begin{equation*}
\begin{split}
    \h{W}_{\tau_2,\tau_3} &= \{i\in V: \forall k~s.t.~i\in V_k,~\exists j\neq i,~\tau_2<|\h{\widetilde{\Theta}}^{(k)}_{i,j}|<\tau_3\}\\
    &\subset\{i\in V: \forall k~s.t.~i\in V_k,~\exists j\neq i,~0<|\widetilde{\Theta}^{(k)}_{i,j}|<\nu\} \\
    &= W_\nu,
\end{split}
\end{equation*}
and 
\begin{equation*}
\begin{split}
    \h{W}_{\tau_2,\tau_3} &\supset\{i\in V: \forall k~s.t.~i\in V_k,~\exists j\neq i,~\psi\leq |\widetilde{\Theta}^{(k)}_{i,j}|\leq \delta-2\varepsilon\} \\
    &= \{i\in V: \forall k~s.t.~i\in V_k,~\exists j\neq i,~0<|\widetilde{\Theta}^{(k)}_{i,j}|< \delta\} \\
    &= W_\delta,
\end{split}
\end{equation*}
where the second line is due to Assumption \ref{assump:app_Schur_distortion2}. Therefore, combining these results with \eqref{eq:app_popEOc}, one has $\h{E}_{O^c} = O^c\cap(\h{W}_{\tau_2,\tau_3}\times \h{W}_{\tau_2,\tau_3}) = \cS_{\off}$.
\end{proof}

\begin{proof}[Proof of Theorem \ref{thm:app_Theta_tilde_est_err}]
Our proof follows similar ideas to the proof of Theorem 4.1 in \citep{vinci2019graph}, while the main difference lies in the correlation matrix estimation. In particular, note that the population MAD$_{\GQ}$ solution $\til{\Theta}$ is the same as the one in \citep{vinci2019graph}, computed from the population correlation matrix. While for the finite sample estimate $\h{\til{\Theta}}$ we defined in \eqref{eq:MADgqLasso}, it depends on the rank based correlation $\h{\Sigma}_O = \h{\Sigma}^{(\tau)}_O$ or $\h{\Sigma}^{(\rho)}_O$, instead of the sample covariance matrix.

However, one thing to note in the proof of Theorem 4.1 in \citep{vinci2019graph} is that we only care about the infinity-norm error bound $\|\h{\Sigma}_O-\Sigma_O\|_{\infty}$ instead of its specific form. The following lemma characterizes this error bound for both rank-based correlation estimates.
\begin{lemma}\label{lem:app_rank_corr_err}
There exists a universal constant $C>0$, such that as long as $n_k\geq C$ holds for all $1\leq k\leq K$, with probability at least $1-\sum_{k=1}^Kp_k^{-10}$, $\widehat{\Sigma}^{(\rho)}_O$ and $\widehat{\Sigma}^{(\tau)}_O$ defined in \eqref{eq:app_Sigma_rank} satisfy
\begin{equation*}
    \|\widehat{\Sigma}^{(\rho)}_O-\Sigma_O\|_{\infty}\leq C\max_k\sqrt{\frac{\log p_k}{n_k}},\quad \|\widehat{\Sigma}^{(\tau)}_O-\Sigma_O\|_{\infty}\leq C\max_k\sqrt{\frac{\log p_k}{n_k}}.
\end{equation*}
\end{lemma}
Lemma \ref{lem:app_rank_corr_err} is proved by applying Theorems 4.1 and 4.2 in \citep{liu2012high} on each block $k$ and a union bound over $K$ blocks.

Inspired by Lemma \ref{lem:app_rank_corr_err}, we specifically change the upper bound $\sigma(\bar{n},p^b)$ in \citep{vinci2019graph} to $C\max_k\sqrt{\frac{\log p_k}{n_k}}$, which is an upper bound for $\|\widehat{\Sigma}_O-\Sigma_O\|_{\infty}$ with probability at least $1-\sum_{k=1}^Kp_k^{-10}$. Similar to \citep{vinci2019graph}, we let $\lambda_{j,l} = \frac{C_0}{\alpha}\max_k\sqrt{\frac{\log p_k}{n_k}}$ with $C_0= 8C$. Recognizing the fact that our sample size condition in Theorem \ref{thm:app_Theta_tilde_est_err} suggests 
$$
C\max_k\sqrt{\frac{\log p_k}{n_k}}\leq \left[6\left(1+\frac{8}{\alpha}\right)^2\til{d}\max\{\kappa_{\til{\Sigma}}\kappa_{\til{\Gamma}}, \kappa_{\til{\Sigma}}^3\kappa_{\til{\Gamma}}^2\}\right]^{-1},
$$
we can follow the same arguments as in \citep{vinci2019graph}, and eventually show that 
\begin{equation}\label{eq:app_MADgqLasso_supp}
    \{(j,l):\h{\til{\Theta}}_{j,l} \neq 0\}\subset \{(j,l):\til{\Theta}_{j,l} \neq 0\},
\end{equation}
and $$\|\h{\til{\Theta}}-\til{\Theta}\|_{\infty}\leq \frac{C_0}{4}\kappa_{\til{\Gamma}}(1+\frac{8}{\alpha})\max_k\sqrt{\frac{\log p_k}{n_k}},$$ with probability at least $1-\sum_{k=1}^Kp_k^{-10}$.
\end{proof}
    
\begin{proof}[Proof of Theorem \ref{thm:app_Theta_Schur_err}]
    To prove Theorem \ref{thm:app_Theta_Schur_err}, we make use of Lemma C.2 and Lemma C.3 in \citep{vinci2019graph}, which bounds the entrywise error bounds of Schur complements as a function of the error of full matrices. For completeness, here we restate these two lemmas.
    \begin{lemma}[Lemma C.2 in \citep{vinci2019graph}]\label{lem:app_schurbnd}
        Let $X,Y\in\mathbb{R}^{p\times p}$ be positive definite matrices. Then, for any nonempty set $A\subset [p]$,
        \begin{equation}
        \left\| X/X_{A^cA^c} -Y/Y_{A^cA^c}\right\|_\infty \le \frac{\lambda_{\rm max}(X)}{\lambda_{\rm min}(X)}\frac{\lambda_{\rm max}(Y)}{\lambda_{\rm min}(Y)}\left\| X-Y\right\|_2
        \end{equation}
        where $X/X_{A^cA^c}=X_{AA}-X_{AA^c}(X_{A^cA^c})^{-1}X_{A^cA}$ is the Schur Complement of the block $A^c\times A^c$ of the matrix $X$, $\lambda_{\rm min}(X)$ and $\lambda_{\rm max}(X)$ are the smallest and the largest eigenvalues of $X$.
    \end{lemma}
    \begin{lemma}[Lemma C.3 in \citep{vinci2019graph}]\label{lem:app_spectral_max}
        For any $p\times p$ symmetric matrix $X$ with max row-degree smaller than or equal to $d$, we have 
        \begin{equation}
        \Vert X\Vert_2\le\min(\sqrt{\Vert X\Vert_0},d)\Vert X\Vert_\infty
        \end{equation}
    where $\Vert X\Vert_2$ is the spectral norm, $\Vert X\Vert_0:=|\{(i,j):X_{ij}\neq 0\}|$,  and $\Vert X\Vert_\infty$ is the max norm.
\end{lemma}
    
 As has been shown in the proof of Theorem \ref{thm:app_Theta_tilde_est_err}, under appropriate conditions, with probability at least $1-\sum_{k=1}^Kp_k^{-10}$, the support set of $\h{\til{\Theta}}$ is a subset of the support set of $\til{\Theta}$; and $\|\h{\til{\Theta}}-\til{\Theta}\|_{\infty}\leq \frac{C_0}{4}\kappa_{\til{\Gamma}}(1+\frac{8}{\alpha})\max_k\sqrt{\frac{\log p_k}{n_k}}$. Therefore, by applying Lemma \ref{lem:app_schurbnd}, we have
    \begin{equation}\label{eq:app_Theta_Schur_err_step}
        \begin{split}
            \|\h{\til{\Theta}}^{(k)}-\til{\Theta}^{(k)}\|_{\infty}\leq \frac{\lambda_{\max}(\til{\Theta})}{\lambda_{\min}(\til{\Theta})}\frac{\lambda_{\max}(\h{\til{\Theta}})}{\lambda_{\min}(\h{\til{\Theta}})}\|\h{\til{\Theta}}-\til{\Theta}\|_2,
        \end{split}
    \end{equation}
    and an application of Lemma \ref{lem:app_spectral_max} leads to 
    \begin{equation*}
        \begin{split}
            \|\h{\til{\Theta}}-\til{\Theta}\|_2&\leq\min\{\sqrt{p+\til{s}}, \til{d}\}\|\h{\til{\Theta}}-\til{\Theta}\|_{\infty}\\
            &\leq \frac{C_0}{4}\kappa_{\til{\Gamma}}(1+\frac{8}{\alpha})\min\{\sqrt{p+\til{s}}, \til{d}\}\max_k\sqrt{\frac{\log p_k}{n_k}}\\
            &\leq \frac{1}{2}\lambda_{\min}(\til{\Theta}),
        \end{split}
    \end{equation*}
    where the last line is due to the sample size condition in Theorem \ref{thm:app_Theta_Schur_err}. Hence
    $\lambda_{\max}(\h{\til{\Theta}})\leq \lambda_{\max}(\til{\Theta}) +\|\h{\til{\Theta}}^{(k)}-\til{\Theta}^{(k)}\|_2\leq \frac{3}{2}\lambda_{\max}(\til{\Theta})$, and $\lambda_{\min}(\h{\til{\Theta}})\geq \lambda_{\min}(\til{\Theta}) -\|\h{\til{\Theta}}^{(k)}-\til{\Theta}^{(k)}\|_2\geq \frac{1}{2}\lambda_{\min}(\til{\Theta})$. Combining these bounds for $\lambda_{\max}(\h{\til{\Theta}})$, $\lambda_{\min}(\h{\til{\Theta}})$, $\|\h{\til{\Theta}}^{(k)}-\til{\Theta}^{(k)}\|_2$, with \eqref{eq:app_Theta_Schur_err_step}, we have
    \begin{equation*}
        \|\h{\til{\Theta}}^{(k)}-\til{\Theta}^{(k)}\|_2\leq \frac{3C_0}{4}\kappa_{\til{\Theta}}^2\kappa_{\til{\Gamma}}(1+\frac{8}{\alpha})\min\{\sqrt{p+\til{s}}, \til{d}\}\max_k\sqrt{\frac{\log p_k}{n_k}}.
    \end{equation*}
\end{proof}

\begin{proof}[Proof of Lemma \ref{lem:app_rank_corr_err}]
    Firstly, we note that our rank-based correlation matrices $\h{\Sigma}^{(\rho)}$ and $\h{\Sigma}^{(\tau)}$ are computed in three steps: we compute the rank statistics, Kendall's tau and Spearman's rho, for each block as in \eqref{eq:app_rho_tau_block}; we then combine the results of $K$ blocks by taking an average on the overlapping part; finally, two sine transformations are applied entrywise on the combined rank statistics as in \eqref{eq:app_Sigma_rank}. In the following, we will prove the error bounds following the same three steps, and leverage similar ideas from the proof of Theorems 4.1 and 4.2 in \citep{liu2012high}. 
    
    In the first step, we note that for $1\leq k\leq K$, $\widehat{\rho}^{(k)}_{j,l}$ and $\widehat{\tau}^{(k)}_{j,l}$ are computed as the standard Spearman's rho and Kendall's tau from the data $X^{(k)}$ of the $k$th block, and hence we can apply the Hoeffding's inequality for U-statistics as in the proof of Theorems 4.1 and 4.2 in \citep{liu2012high}:
    \begin{equation}\label{eq:app_rho_tau_err}
        |\h{\rho}^{(k)}_{j,l} - \bE\h{\rho}^{(k)}_{j,l}|\leq C\sqrt{\frac{\log p_k}{n_k}},\quad|\h{\tau}^{(k)}_{j,l} - \bE\h{\tau}^{(k)}_{j,l}|\leq C\sqrt{\frac{\log p_k}{n_k}},
    \end{equation}
    for $j,l\in[p_k]$ and a universal constant $C>0$, with probability at least $1-p_k^{-10}$. Here, the exponent $-10$ can also be changed by any negative constant $-c$, as long as $C>0$ in the upper bound in \eqref{eq:app_rho_tau_err} is chosen appropriately. By the definition of $\h{\rho}$ and $\h{\tau}$ in \eqref{eq:app_rho_tau} and a union bound over $1\leq k\leq K$, we know that with probability at least $1-\sum_{k=1}^Kp_k^{-10}$, 
    \begin{equation*}
        |\h{\rho}_{j,l}-\bE\h{\rho}_{j,l}|\leq C\max_k\sqrt{\frac{\log p_k}{n_k}},\quad |\h{\tau}_{j,l}-\bE\h{\tau}_{j,l}|\leq C\max_k\sqrt{\frac{\log p_k}{n_k}}.
    \end{equation*}
    
    To transform the bounds above to the bounds for the estimated correlation matrices, we note that $\Sigma_{j,l}=\sin\left(\frac{\pi}{2}\bE\h{\tau}_{j,l}\right)$, $\h{\Sigma}^{(\tau)}_{j,l}=\sin\left(\frac{\pi}{2}\h{\tau}_{j,l}\right)$, and the Lipschitz property of the sine function implies
    \begin{equation}\label{eq:app_Sigma_tau_err}
        |\h{\Sigma}^{(\tau)}_{j,l}-\Sigma^{(\tau)}_{j,l}|\leq \frac{\pi}{2}|\h{\tau}_{j,l}-\bE\h{\tau}_{j,l}|.
    \end{equation}
    While for $\h{\Sigma}^{(\rho)}$, as has been shown in the proof of Theorem 4.1 in \citep{liu2012high}, 
    \begin{equation*}
        \bE\h{\rho}^{(k)}_{j_k,l_k} = \frac{6}{\pi}\arcsin\left(\frac{\Sigma_{j,l}}{2}\right)+\frac{6}{\pi(n_k+1)}\left(\arcsin(\Sigma_{j,l}) - 3\arcsin\left(\frac{\Sigma_{j,l}}{2}\right)\right).
    \end{equation*}
    Denote by $a^{(k)}_{j,l}$ the excess error term:  $$a^{(k)}_{j,l}=\bE\h{\rho}^{(k)}_{j_k,l_k}-\frac{6}{\pi}\arcsin\left(\frac{\Sigma_{j,l}}{2}\right) = \frac{6}{\pi(n_k+1)}\left(\arcsin(\Sigma_{j,l}) - 3\arcsin\left(\frac{\Sigma_{j,l}}{2}\right)\right),$$ then we know that $|a^{(k)}_{j,l}| \leq \frac{12}{n_k+1}$. Thus, by the definition of $\widehat{\rho}$ in \eqref{eq:app_rho_tau}, 
    \begin{equation*}
         \bE\h{\rho}_{j,l} = \frac{6}{\pi}\arcsin\left(\frac{\Sigma_{j,l}}{2}\right)+\frac{\sum_{k=1}^K\ind{j,l\in V_k}a^{(k)}_{j,l}}{\sum_{k=1}^K\ind{j,l\in V_k}},
    \end{equation*}
    and 
    \begin{equation}\label{eq:app_Sigma_rho_err}
    \begin{split}
        |\h{\Sigma}^{(\rho)}_{j,l}-\Sigma^{(\rho)}_{j,l}| &= \left|2\sin\left(\frac{\pi}{6}\h{\rho}_{j,l}\right)-2\sin\left(\frac{\pi}{6}\bE\h{\rho}_{j,l}-\frac{\sum_k\ind{j,l\in V_k}a^{(k)}_{j,l}}{\sum_k\ind{j,l\in V_k}}\right)\right|\\
        &\leq \frac{\pi}{3}\left(|\h{\rho}_{j,l}-\bE\h{\rho}_{j,l}|+\frac{12}{\min_kn_k+1}\right).
    \end{split}
    \end{equation}
    Therefore, combining \eqref{eq:app_Sigma_rho_err}, \eqref{eq:app_Sigma_tau_err}, \eqref{eq:app_rho_tau_err}, the proof of Lemma \ref{lem:app_rank_corr_err} is completed.
\end{proof}

\end{document}